\documentclass[american,aps,pra,reprint,superscriptaddress,longbibliography]{revtex4-1}

\usepackage{amsmath,amssymb}

\usepackage{bbm}
\usepackage{color}
\usepackage{graphicx}
\usepackage[american]{babel}
\usepackage[utf8]{inputenc}
\usepackage{times}
\usepackage{epstopdf}
\usepackage{braket} 				
\usepackage{amsthm}

\newcommand{\ketbra}[2]{|#1\rangle\!\langle#2|}

\definecolor{mygrey}{gray}{0.35}
\definecolor{myblue}{rgb}{0.2,0.2,0.8}
\definecolor{myzard}{cmyk}{0,0,0.05,0}
\definecolor{mywhite}{rgb}{1,1,1}
\definecolor{myred}{rgb}{0.9,0.1,0.}
\usepackage[colorlinks=true,citecolor=myblue,linkcolor=myblue,urlcolor=myblue]{hyperref}

\newtheoremstyle{customStyle1}  
{0pt}       
{0pt}       
{\normalfont}   
{\parindent}        
{\em}  
{. --}   	 
{.5em}       
{\thmname{#1}\thmnumber{ #2}\thmnote{ (#3)}}  

\theoremstyle{customStyle1}

\newtheorem{theorem}{Theorem}
\newtheorem{lemma}[theorem]{Lemma}
\newtheorem{definition}[theorem]{Definition}
\newtheorem{corollary}[theorem]{Corollary}

\newtheorem{proposition}[theorem]{Proposition}

\ifx\proof\undefined
\newenvironment{proof}[1][\protect\proofname]{\par
	\normalfont\topsep6\p@\@plus6\p@\relax
	\trivlist
	\itemindent\parindent
	\item[\hskip\labelsep\scshape #1]\ignorespaces
}{%
\endtrivlist\@endpefalse
}
\providecommand{\proofname}{Proof}
\fi

\newtheorem*{theorem*}{Theorem}
\newtheorem*{lemma*}{Lemma}
\newtheorem*{corollary*}{Corollary}
\newtheorem*{observation*}{Observation}
\newtheorem*{proposition*}{Proposition}

\newcommand{\tr}{\operatorname{\bf{tr}}} 
\newcommand{\mini}[1]{\underset{#1}{\operatorname{\bf{min}}}} 
\newcommand{\maxi}[1]{\underset{#1}{\operatorname{\bf{max}}}} 
\newcommand{\infi}[1]{\underset{#1}{\operatorname{\bf{inf}}}} 
\newcommand{\supr}[1]{\underset{#1}{\operatorname{\bf{sup}}}} 
\newcommand{\id}{\mathbbm{1}}

\let\oldforall\forall
\renewcommand{\forall}{\quad \oldforall}

\begin{document}
\title{Quantifying dynamical coherence with dynamical entanglement }

\author{Thomas Theurer}
\thanks{These two authors contributed equally}
\affiliation{Institute of Theoretical Physics and IQST, Universität Ulm, Albert-Einstein-Allee
	11, D-89069 Ulm, Germany}

\author{Saipriya Satyajit}
\thanks{These two authors contributed equally}
\affiliation{Indian Institute of Technology Bombay, Mumbai 400076, India}

\author{Martin B. Plenio}
\affiliation{Institute of Theoretical Physics and IQST, Universität Ulm, Albert-Einstein-Allee
	11, D-89069 Ulm, Germany}

\begin{abstract}
	Coherent superposition and entanglement are two fundamental aspects of non-classicality. Here we provide a quantitative connection between the two on the level of operations by showing that the dynamical coherence of an operation upper bounds the dynamical entanglement that can be generated from it with the help of additional incoherent operations. 
	In case a particular choice of monotones based on the relative entropy is used for the quantification of these dynamical resources, this bound can be achieved. In addition, we show that an analog to the entanglement potential exists on the level of operations and serves as a valid quantifier for dynamical coherence. 
\end{abstract}

\maketitle


{\em Introduction.} -- 
In the last decades growing evidence has emerged that quantum technologies are able to outperform their classical counterparts, e.g., in communication~\cite{Bennett1984,Gisin2002,Gisin2007} and computation~\cite{Deutsch1985,Nielsen2010,Watrous2018}, but also in sensing~\cite{Degen2017} and metrology~\cite{Wineland1992,Huelga1997}. These operational advantages originate from non-classical traits of quantum physics which are thus considered resources. One of the most important example of such a resource is entanglement~\cite{Plenio2007,Horodecki2009} {which describes correlations between spatially separated systems that are without classical equivalent. Yet, there exist situations in which it is neither natural nor sufficient to describe non-classicality with entanglement alone. This is for example the case if one considers non-composite systems which have no natural concept of locality, whence superposition is considered to be a quantum resource~\cite{Sudarshan1963,Glauber1963a,Glauber1963b}. 
It is important to understand if and how these different notions of non-classicality are connected. As shown in Refs.~\cite{Kim2002,Xiang-bin2002}, optical non-classicality is a prerequisite for the creation of entanglement by beam splitters, and there exists a simple relation between squeezing and the entanglement that can be generated from it by passive optical elements~\cite{Wolf2003}. This finding was used in Ref.~\cite{Asboth2005} to measure optical non-classicality via its entanglement potential, i.e., the amount of two-mode entanglement that can be generated from the field using passive linear optics, auxiliary classical states, and ideal photo detectors. In general, local superposition can be faithfully converted into (multilevel) entanglement using only operations that cannot create these  superpositions~\cite{Vogel2014,Killoran2016,Theurer2017,Regula2018}. In a similar spirit, the activation of coherence and discord into entanglement was studied, e.g., in Refs.~\cite{Piani2011,Gharibian2011,Streltsov2011}. Unified approaches to these three resources were recently presented in Refs.~\cite{Egloff2018,Zhou2019}.

Resource theories have proven beneficial for the systematic study of the various notions of non-classicality present in quantum states~\cite{Vedral1997,Horodecki2003,Aberg2006,Gour2008,Horodecki2013,Brandao2013,Grudka2014,Baumgratz2014,DelRio2015}, partly because they allow to quantify resources in an objective manner. This quantification is achieved with the help of resource measures that satisfy physically motivated constraints such as monotonicity under operations that cannot create the investigated resource. Using such resource measures, it was shown in Ref.~\cite{Streltsov2015} that local superposition in the form of coherence has not only a qualitative, but also a quantitative connection to entanglement:  
for a large class of commonly used resource measures, the amount of entanglement that can be generated from a local state with the help of incoherent operations is upper bounded by the coherence of that state and, for two specific measures, the two quantities coincide. Moreover, also in the case of coherence theory, the entanglement potentials (as measured with any entanglement measure) can be used to define valid coherence measures. 
These results where expanded in Refs.~\cite{Ma2016,Zhu2017PRA,Zhu2017,Zhu2018,Regula2018,Ren2020}.
In this work, we generalize the findings of Ref.~\cite{Streltsov2015} to \textit{operations}. 

As mentioned above, a resource theory of quantum states allows to describe the resources present in \textit{states} in an operationally meaningful way. However, if we speak about operational advantages granted by quantum technologies, we intend to perform certain \textit{tasks} better than it is possible with classical devices. To do this, static resources in the form of resource states have to be converted into dynamical resources by combining them with free operations. Thus, ultimately, we are interested in the quantification of dynamical resources in the form of quantum operations~\cite{Theurer2019}. For this, often quantities such as the resource generation capacity, i.e., the achievable increase in static resources, or the resource cost, i.e., the minimal amount of static resources necessary to simulate the operation in combination with free operations, are employed~\cite{Eisert2000,Collins2001,Nielsen2003,Bennett2003,Mani2015,Xi2015,Garcia2016,Bu2017,Dana2017}. Yet, this approach cannot be used to quantify all relevant properties of quantum operations, e.g.,  their ability to detect coherence~\cite{Theurer2019}. This is one of the reasons why resource theories of operations have been considered recently~\cite{Coecke2016,Zhuang2018,Theurer2019,Yuan2019,Oszmaniec2019,Takagi2019,Wang2019,Liu2020,Liu2019,Seddon2019,Gour2019a,Bauml2019,Gour2019b,Wang2019b,Saxena2019,Gonda2019,Li2020,Takagi2020,Hsieh2020}. As described above, another reason is that it seems to be natural to quantify the value of operations directly without a detour via the quantification of states and finally, since states can be identified with their preparation channels, resource theories of operations are a unifying framework. 

Our work begins with an introduction of the basic framework of dynamical resources, followed by our main results, their discussion, and a conclusion. Detailed proofs are deferred to the Supplemental Material (SM)~\cite{SM}.

\begin{figure}[tb!]
	\centering
	\includegraphics[width=1\linewidth]{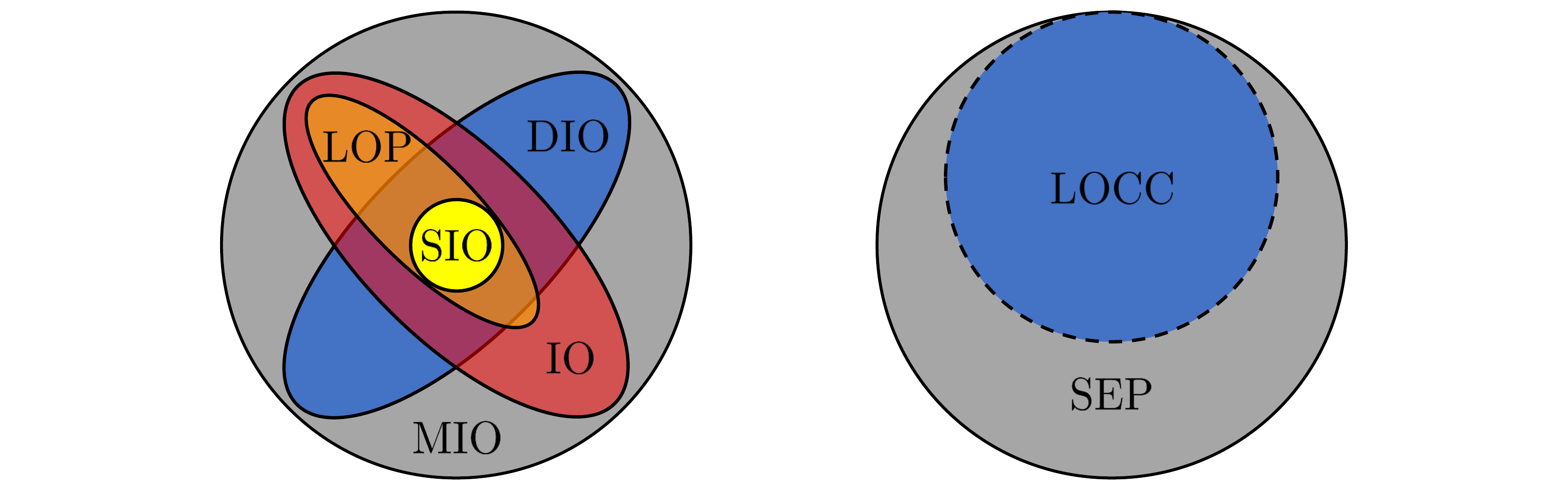}
	\caption{{\bf Inclusion relations of the considered operations.} On the left, different sets of  operations considered free in coherence theories as well as their relations are shown and on the right, we show sets of operations that cannot create entanglement. Since LOCC is not closed, we show its boundary with a dashed line. The closure of LOCC, which we denote by $\rm\overline{LOCC}$, also includes this boundary.}
	\label{fig::VennDiagrams}
\end{figure}

{\em Basic framework.} --
In the following, we denote by $\mathcal{D}_A/\mathcal{I}_A$ the set of quantum states/incoherent states on system $A$ and by $\mathcal{W}_{A|B}$ the set of separable states with respect to a bi-partition into $A$ and $B$. For ease of notation, we define the incoherent states as those diagonal in the computational basis~$\left\{\ket{i}\right\}$. For composite systems, the incoherent pure states are the tensor products of the incoherent pure states of the subsystems. To denote quantum operations, i.e., completely positive and trace preserving (CPTP) maps, we use large Greek letters and, if necessary, denote the parties on which they act by superscripts, i.e., $\Theta^{AB}$ represents a CPTP map acting on $AB$. In principle, such an operation can have different input and output dimensions, therefore $\Theta^{A_{\rm out}B_{\rm out}\leftarrow A_{\rm in}B_{\rm in}}$ would be more appropriate. However, to avoid an unnecessarily lengthy notation, we will not write down the inputs and outputs explicitly but instead, whenever we concatenate operations, demand implicitly that the dimensions fit. For the same reason, we suppress the superscripts if they are clear from the context. 
After these introductory comments, we move to a precise definition of dynamical coherence and entanglement. To begin with, we list sets of operations that, in their relevant contexts, are considered to be free of dynamical coherence or entanglement. 
We depict them in Fig.\ref{fig::VennDiagrams}, for overviews see Refs.~\cite{,Plenio2007,Horodecki2009,Streltsov2017,Chitambar2019}. 

The set of maximally incoherent operations (MIO)~\cite{Aberg2006,Liu2017,Diaz2018} is the maximal set of operations that maps $\mathcal{I}$ to itself. If such an operation can be decomposed into Kraus operators that preserve $\mathcal{I}$ individually, it is called an incoherent operation (IO)~\cite{Baumgratz2014}. The subset of MIO that cannot make use of coherence is called DIO (dephasing-covariant incoherent operations ~\cite{Meznaric2013,Liu2017,Chitambar2016,Marvian2016,Chitambar2016b}) and if this holds again individually for any element of a Kraus decomposition, we speak about strictly incoherent operations (SIO)~\cite{Winter2016}. The set LOP is a superset of SIO and a subset of IO, and denotes the set of free operations (on the wire) in the framework of local operations and physical wires~\cite{Egloff2018}. 
The total dephasing operation (with respect to the incoherent basis) will be denoted by $\Delta$ and is contained in all of these sets. In addition, all of these sets can exactly prepare the states within $\mathcal{I}$, which we therefore call free in this context. 

Moving on to entanglement, we consider the set of local operations and classical communication (LOCC)~\cite{Bennett1996}, the closure of LOCC ($\rm{\overline{LOCC}}$)~\cite{Chitambar2014}, and the set of separable operations (SEP)~\cite{Rains1997,Vedral1998}, which are the maximal set of operations that map separable states to separable states in a complete sense, i.e., even if they are applied to subsystems. Similarly, the states that can be prepared by these sets of operations, i.e., their free states, are exactly the separable ones. 

Each set of operations listed above leads to a different dynamical resource theory (and also to a different static one).   As detailed in the corresponding references, depending on the context, there are valid arguments to consider each of these sets of operations as free. Once the choice of free operations is made, all operations that are not contained in it are called dynamically coherent or entangled \textit{with respect to this specific choice of free operations}.
Our results are largely independent of the specific choices, which shows that there is a deep connection between dynamical coherence and entanglement. Hence, we use $\mathcal{C}$ to represent either SIO, LOP, IO, DIO, or MIO.  Similarly, $\mathcal{E}$ is a placeholder for $\rm LOCC, {\rm\overline{LOCC}}$, or  $\rm{SEP}$. To make clear which partition we are considering, we write, e.g., $\mathcal{E}^{A|B}$ to denote a bi-partition into $A$ and $B$. 

To quantify dynamical resources, we use functions $D(\rho,\sigma)$ (on quantum states $\rho$, $\sigma$) that are jointly convex, contractive under CPTP maps, and zero if and only if $\rho$ is equal to $\sigma$. We call these functions divergences. 
Note that the relative entropy $S(\rho,\sigma)=\tr \left(\rho \log \rho \right))-\tr \left(\rho \log \sigma\right)$, the trace distance $\|\rho-\sigma \|_1$, and many R\'enyi  entropies qualify as divergences. To include destructive measurements into our framework, we associate with positive operator-valued measures (POVMs) given by elements $\Pi_n$ the CPTP maps
\begin{align}\label{eq:CPTPfromPOVM}
	M(\rho)=\sum_n \tr(\Pi_n \rho) \ketbra{x_n}{x_n},
\end{align}
where the states $\ket{x_n}$ are orthonormal. 
We then consider a destructive measurement in $\mathcal{E}$ iff the associated CPTP map with $\ket{x_n}=\ket{n}^A\otimes \ket{n}^B$ is in $\mathcal{E}$. Analogously, with $\ket{x_n}=\ket{n}$, we define destructive measurements in $\mathcal{C}$. 
For quantum instruments $I$ allowing us to do subselection according to a variable $n$, i.e., we obtain with probability $p_n=\tr(\Lambda_n(\rho))$ an output $\rho_n=\Lambda_n(\rho)/p_n$, we use the same construction to define a CPTP map
\begin{align}\label{eq:CPTPfromInstrument}
\tilde{I}(\rho)=\sum_n \Lambda_n(\rho)\otimes\ketbra{x_n}{x_n}.
\end{align}
Treating subselection in this way, we can reduce our analysis to trace preserving operations, since now it is always possible to implement the subselection at a later point with a free measurement. 
This has the additional advantage that the ability to apply or not apply subselection according to a specific variable, which depends on the precise circumstances under which an experiment is realized, has a direct reflection in our framework. 

Let $D$ be a divergence and $\mathcal{S}$ either $\mathcal{D}$ or a subset thereof. As described in Refs.~\cite{Cooney2016,Leditzky2018,Gour2019c,Liu2020,Liu2019}, one can then define the following quantities on operations $\Theta^A,\Lambda^A$:
\begin{align}
D^\mathcal{S}(\Theta, \Lambda)=\supr{\sigma \in \mathcal{S}}\ D \left( \left(\Theta^A\otimes\id^E\right)\sigma, \left(\Lambda^A\otimes\id^E\right)\sigma \right), 
\end{align}
where the optimization over states is understood to include an optimization over different dimensions of an auxiliary system $E$ and (as explained above) $A$ denotes the combination of the (potentially composite) input and output systems of $\Theta$ and $\Lambda$.
Analogously, we also define their measured versions as
\begin{align}
&D^{\mathcal{S}, \mathcal{M}}(\Theta, \Lambda) =\supr{M \in \mathcal{M}}\supr{\sigma \in \mathcal{S}}\ D \left( M\left(\Theta\otimes\id\right)\sigma, M\left(\Lambda\otimes\id\right)\sigma \right), \nonumber
\end{align}
where $\mathcal{M}$ denotes the set of CPTP maps associated with a set of POVMs as defined in of Eq.~(\ref{eq:CPTPfromPOVM}). In this work, we are mainly interested in the case where this is either the set of free destructive measurements ($\mathcal{M}={\rm free}$) within a given resource theory or the set of all destructive measurements ($\mathcal{M}={\rm all}$). With a bit of abuse of notation, we also write $D^{\mathcal{S},{\rm no}}$ for $D^\mathcal{S}$ from here on, indicating that no measurement was included. Let us note here the well known fact that the supremum over the states is always achievable for the dimension of the auxiliary space equal to the input dimension of $\Theta$ 
(which is a simple consequence of joint convexity of $D$ and the Schmidt decomposition). Note also that there exist examples of $D$ for which $D^{\mathcal{S},{\rm no}}=D^{\mathcal{S},{\rm all}}$, e.g., if $D$ equals the trace distance, while for others, such as $D$ equaling the relative entropy, this is not true: the measured relative entropy is equal to the relative entropy if and only if the two arguments commute~\cite{Fuchs1996}. We will see later why these two examples are of special interest to us. 
Motivated by Refs.~\cite{Liu2020,Liu2019,Gour2019a}, the quantities introduced above allow us to define the following functionals.
\begin{definition}\label{def:Meas}
	Let $\mathcal{M}\in\{\rm free, all, no\}$, where $\rm free$ denotes the set of destructive measurements which are free \textit{within the set of operations $\mathcal{E}$}. Then, for a divergence $D$ and $\mathcal{S} \in \left\{\mathcal{D},\mathcal{W}_{AE_A|BE_B}\right\}$, we define
	\begin{align}
		&E_{\mathcal{E}^{A|B},D}^{\mathcal{S},\mathcal{M}}\left(\Theta\right):=\infi{\Lambda\in \mathcal{E}^{A|B}}\ D^{\mathcal{S},\mathcal{M}}(\Theta, \Lambda). \nonumber
	\end{align}
	In complete analogy, we define
	\begin{align}
		&C_{{\mathcal{C}},D}^{\mathcal{S},\mathcal{M}}\left(\Theta\right):=\infi{\Lambda\in \mathcal{C}}\ D^{\mathcal{S},\mathcal{M}}(\Theta, \Lambda) \nonumber
	\end{align} 
	for $\mathcal{S}\in\{\rm \mathcal{D},\mathcal{I}\}$ and $\mathcal{M}\in \{\rm free,all, no \}$.
\end{definition}
As we will show now, the functionals defined above are so-called resource measures and monotones on the level of operations, which was also partially proven before in Refs.~\cite{Liu2020,Liu2019,Gour2019a} (for the cases where no measurements are included).
We call a functional $F$ from quantum operations to the non-negative real numbers a resource monotone if $F$ is monotonic under concatenation with free operations, i.e., if $F(\Theta)\ge F\left(\Phi_2\left(\Theta\otimes\id\right)\Phi_1\right)$ for all  $\Phi_i$ that are free within the resource theory and for all $\Theta$ CPTP. Due to the construction in Eq.~(\ref{eq:CPTPfromInstrument}), we will assume that also the $\Phi_i$ are deterministic. 
If $F$ is in addition faithful, i.e., zero if and only if $\Theta$ is free, we call it a resource measure.

\begin{proposition}\label{prop:monotones}
	Let $\mathcal{S}$ denote either the set of all quantum states or the set of free quantum states, $\mathcal{M} \in  \left\{ \rm{all,free, no} \right\}$, $\tilde{\mathcal{M} }\in  \left\{ \rm{all, no} \right\}$, $\tilde{\mathcal{E}} \in \left\{\rm{\overline{LOCC}}, \rm{SEP} \right\}$, and $\tilde{\mathcal{C}} \in \{\rm LOP, IO ,MIO \}$.
	Then 
	\begin{align*}
		E_{\tilde{\mathcal{E} }^{A|B},D}^{\mathcal{S},\mathcal{M}}\left(\Theta\right),\ C_{{\rm MIO},D}^{\mathcal{S},\mathcal{M}}\left(\Theta\right),\ 
		C_{{\mathcal{C}},D}^{\mathcal{D},\tilde{\mathcal{M}}}\left(\Theta\right),\ \text{and}\ C_{{\tilde{\mathcal{C} }},D}^{\mathcal{D},{\rm free}}\left(\Theta\right)
	\end{align*}
	are convex resource measures. The remaining functionals from Def.~\ref{def:Meas}  are convex resource monotones, with those defined via destructive DIO/SIO measurements vanishing on all operations.
\end{proposition}

\begin{figure}[tb!]
	\centering
	\includegraphics[width=\linewidth]{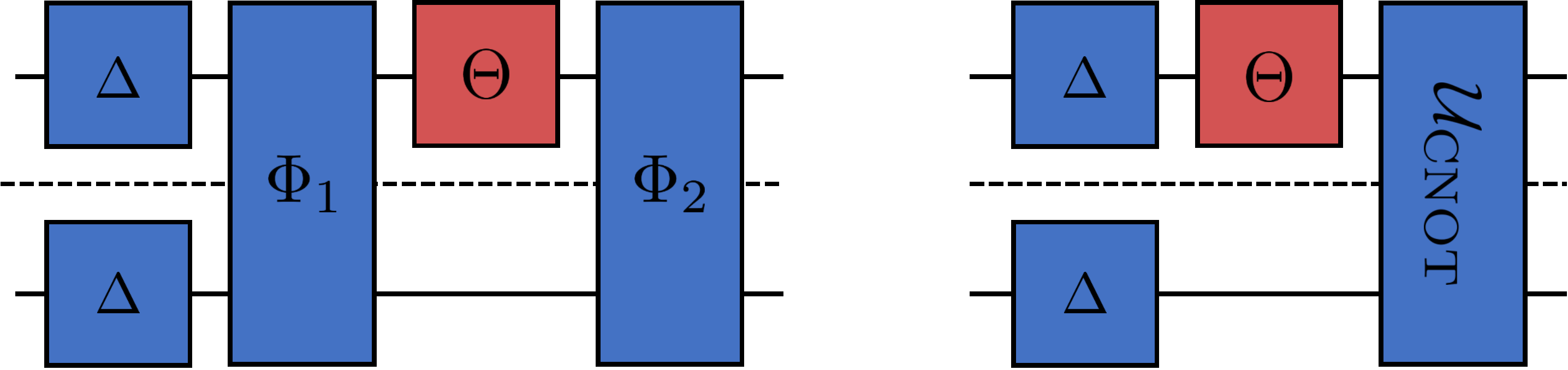}
	\caption{{\bf Converting coherence to entanglement.} The amount of dynamical entanglement created by the setup on the left (where $\Phi_i$ are operations in $\mathcal{C}$) is upper bounded by the coherence of the operation $\Theta$, if related monotones are used for the quantification (see Thm.~\ref{thm:inequality}). The dashed line represents the spatial separation with respect to which we take our bi-partition into parties $A$ and $B$, which are represented by the solid lines. On the right, we show an optimal setup in case we are considering two relative entropy based measures (see Thm.~\ref{thm:main}).}
	\label{fig::generalScheme}
\end{figure}
A detailed discussion of the relation of these monotones is provided in the SM.
Of special interest is the case of $D$ equal to the trace distance. Then, the measures have a direct operational interpretation in the single-shot regime: they are proportional to the best bias achievable in the guessing game where one has to distinguish the given operation from the least distinguishable free operation~\cite{Matthews2009,Lami2018} (with the help of the states $\mathcal{S}$ and the outcomes of the measurements $\mathcal{M}$). Therefore, if we consider the trace distance based measure with free states and free destructive measurements, this represents the usefulness of the operation under consideration within the given resource theory: an operation that is barely distinguishable from a free operation using other free operations and states can only lead so a very small operational advantage~\cite{Smirne2017,Theurer2019,Milz2019}, which is the reason why we focused on destructive measurements in Def.~\ref{def:Meas}. An example of such a measure is the NSID measure considered in Ref.~\cite{Theurer2019}. 

{\em Main results.} --
We begin by showing that the dynamical coherence with respect to $\mathcal{C}$ upper bounds the dynamical entanglement with respect to $\mathcal{E}$ that can be generated from it using the setup depicted in Fig.~\ref{fig::generalScheme} on the left, where $\Theta$ is the operation under investigation and  $\Phi_i$ operations in $\mathcal{C}$. 
\begin{theorem}\label{thm:inequality}
	Let $\Phi_i\in \mathcal{C}$. Then	
	\begin{align}
		&C_{{\mathcal{C}},D}^{\mathcal{I},{\rm no}}\left(\Theta\right)   \ge E_{{\mathcal{E}}^{A|B},D}^{\mathcal{W},{\rm no}}\left(\Phi_2\left(\Theta^A\otimes \id^B\right)\Phi_1\Delta\right) \nonumber
	\end{align}
	and
	\begin{align}
		&C_{{\mathcal{C}},D}^{\mathcal{I},{\rm all}}\left(\Theta\right)   \ge E_{{\mathcal{E}}^{A|B},D}^{\mathcal{W},{\rm all}}\left(\Phi_2\left(\Theta^A\otimes \id^B\right)\Phi_1\Delta\right). \nonumber
	\end{align}
\end{theorem}

As stated in the following Theorem, the (generalized) controlled NOT operation, which acts on two subsystems of equal dimension $d$ as 
\begin{align}
	U_\mathrm{CNOT}=\sum_{i}\sum_{j} \ketbra{i}{i}\otimes\ketbra{\mathrm{mod}(i+j,d)}{j},
\end{align}
allows to attain this bound.
Obviously $U_\mathrm{CNOT}$ is unitary (we denote the corresponding CPTP map by $\mathcal{U}_\mathrm{CNOT}$) and contained in SIO.
\begin{theorem}\label{thm:inverse}
	For $S(\rho,\sigma)$ the relative entropy and $\mathrm{dim} B=\mathrm{dim} A$,
	\begin{align}
		C_{{\mathcal{C}},S}^{\mathcal{I},{\rm no}}\left(\Theta\right)   \le E_{{\mathcal{E}}^{A|B},S}^{\mathcal{W},{\rm no}}\left(\mathcal{U}_\mathrm{CNOT}\left(\Theta^A\otimes \id^B\right)\Delta\right) \nonumber
	\end{align}
	holds.
\end{theorem} 
Combining Thms.~\ref{thm:inequality} and~\ref{thm:inverse}, we arrive at one of our main results.
\begin{theorem}\label{thm:main}
	For $S(\rho,\sigma)$ the relative entropy,
	\begin{align}
		C_{{\mathcal{C}},S}^{\mathcal{I},{\rm no}}\left(\Theta\right)   = \supr{\Phi_1,\Phi_2 \in \mathcal{C}}\  E_{{\mathcal{E}}^{A|B},S}^{\mathcal{W},{\rm no}}\left(\Phi_2\left(\Theta^A\otimes \id^B\right)\Phi_1\Delta\right). \nonumber
	\end{align}
	The supremum is achieved for $\mathrm{dim} B=\mathrm{dim} A$, $\Phi_1=\id$ and $\Phi_2=\mathcal{U}_\mathrm{CNOT}$. 
\end{theorem}
This Theorem shows that dynamical entanglement is intimately connected to dynamical  coherence: the coherence is in one-to-one correspondence to the entanglement that can be generated from it by the protocol depicted in Fig.~\ref{fig::generalScheme} on the left, which  only involves auxiliary operations that are free from the coherence perspective. Moreover, as shown in Fig.~\ref{fig::generalScheme} on the right, the optimal generation scheme does not require a pre-processing and only a fixed post-processing. This should be compared to Refs.~\cite{Piani2011,Gharibian2011,Streltsov2011, Streltsov2015}, where the controlled NOT operation plays a central role too.
As a Corollary, we find
\begin{corollary}\label{cor:conversion}
	An operation $\Theta$ can be converted to an operation outside $\mathcal{E}$ with operations in $\mathcal{C}$ if and only if $\Theta$ is not in $\rm MIO$. 
\end{corollary}

\begin{figure}[tb!]
	\centering
	\includegraphics[width=\linewidth]{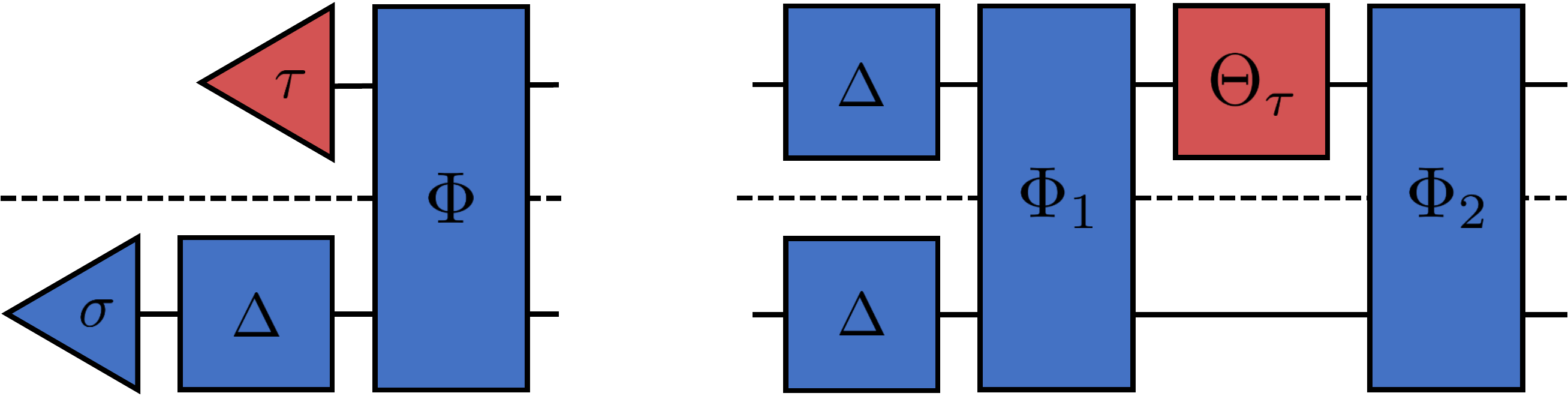}
	\caption{{\bf Reduction to states.} 
	Assume that $\Phi$ is an incoherent operation. With the setup depicted on the left, it is then possible to create static entanglement if and only if the quantum state $\tau$ is coherent~\cite{Streltsov2015}. As discussed in more detail in the SM, for the preparation and replacement channels $\Theta_{\tau}$ with output $\tau$, our setup depicted on the right includes this result on resource theories of states. 
	}
	\label{fig::StateReduction}
\end{figure}
As detailed in the SM and shown in Fig.~\ref{fig::StateReduction}, if $\Theta_\tau$ is a preparation or, more general, a replacement channel with output $\tau$, i.e., $\Theta_\tau \rho=\tau \tr \rho$, by identifying $\tau$ with $\Theta_\tau$, we recover the analog results of the above Theorems and Corollary for resource theories of states which were presented in Ref.~\cite{Streltsov2015}.

{\em Measuring dynamical coherence with dynamical entanglement.} --
Above, we discussed how monotones for dynamical coherence and entanglement bound each other. In this section, we take a complementary approach and define coherence monotones with the help of entanglement monotones. 
\begin{theorem}
	Let $E_{\mathcal{E}^{AE_A|B}}$ be a (convex) resource monotone with respect to the set of operations $\mathcal{E}^{AE_A|B}$. Then 
	\begin{align}\label{eq:defCohWithEnt}
	C_{\mathcal{C}}^{\mathcal{E},E}\left(\Theta\right):=\supr{\Lambda_i\in \mathcal{C}}\ E_{\mathcal{E}^{AE_A|B}}\left(\Lambda_2\left(\Theta^{A} \otimes \id^{E_AB} \right)\Lambda_1\Delta\right) 
	\end{align}
	is a (convex) resource monotone with respect to the operations $\mathcal{C}$. If $E_{\mathcal{E}^{AE_A|B}}$ is in addition faithful, Eq.~(\ref{eq:defCohWithEnt}) defines a measure with respect to $\rm MIO$.
\end{theorem}
This shows that an analog of the entanglement potential discussed for states in Refs.~\cite{Asboth2005,Streltsov2015} also exists on the level of operations: the maximal amount of entanglement that can be generated from an operation by the method depicted in Fig.~\ref{fig::generalScheme} using an additional local auxiliary system serves as a valid measures for the coherence of this operation. 

{\em Discussion.} --  In all the above Theorems, we used conversion schemes as depicted in Fig.~\ref{fig::generalScheme} where initially, we apply a global total dephasing operation.  Whilst one might think that it may be more elegant or natural not to include the total dephasing operation $\Delta$, to our knowledge, it is unavoidable. Moreover, implicitly, it is also necessary in the state case presented in Ref.~\cite{Streltsov2015}: initially, the auxiliary system has to be in an incoherent state, which can be enforced by putting a $\Delta$ as shown on the left of Fig.~\ref{fig::StateReduction}. The same holds true if one considers the transformation of optical non-classicality into entanglement with a beam splitter (or, more generally, a passive linear optics network): the non-classicality of the initial state is only converted faithfully into entanglement if the other input port of the beam splitter is connected to a classical state. 

{\em Conclusions.} --
In this work, we quantitatively connected dynamical entanglement and coherence. 
Our findings not only uncover how two of the most fundamental non-classical traits of quantum mechanics are connected on the level of operations, which is of foundational interest, but also allow to apply findings from the emerging resource theories of dynamical entanglement and coherence to the respective other theory. On a more practical level, our results shed new light on the resources required to obtain operations outside of LOCC. Such operations are a necessary prerequisite to obtain operational advantages in quantum communication. In particular, it might be of interest to apply our findings to quantum key distribution. As shown in Ref.~\cite{Shor2000}, it is not necessary to create entanglement for secret key distribution, but it is sufficient to use channels that could create entanglement in principle, i.e., operations which posses dynamical entanglement. Is it then possible to, e.g., bound key rates with measures of dynamical entanglement? And, using our results, can this be connected to the encoding in non-orthogonal states, i.e., to dynamical coherence? These questions are subject to future investigation.
Moreover, our findings can help to uncover the origin of operational advantages in quantum computation, where entanglement and coherence are widely believed to play a central role. 
As we argued in the introduction, it is natural to investigate the relevance of these resources from a dynamical perspective. Our results suggest that one might want to focus onto dynamical coherence, since it is equivalent to the dynamical entanglement that can be generated from it by a controlled NOT operation, which is frequently used in various quantum algorithms. An improved understanding of the resources responsible for operational advantages in quantum computation in turn will allow for a more systematic construction of quantum algorithms.

\begin{acknowledgments}
	We acknowledge helpful discussions with Ho-Joon Kim, Michele Masini, Ludovico Lami, Dario Egloff, and Sai Vinjanampathy. TT and MBP acknowledge financial support by the ERC Synergy Grant BioQ (grant no 319130). SS acknowledges DAAD WISE and the IQST for financial support of a visit to Ulm University during which this project was initiated.
\end{acknowledgments}

%

\newpage
\clearpage

\onecolumngrid

	\section*{Supplemental Material: Quantifying dynamical coherence with dynamical entanglement}

	In this Supplemental Material we give the proofs of the results presented in the main text and some further details. In particular, we discuss how the monotones introduced in the main text bound each other and how we recover the results of Ref.~\cite{Streltsov2015} for the special case of preparation channels.

	\section{Technical Lemmas}\label{sec:lemmas}
	\setcounter{theorem}{7}
	Here we present three Lemmas that we will use in Sec.~\ref{sec:proofs} to prove our results presented in the main text. The first two will allow us to conclude that the auxiliary systems of some of the monotones introduced in the main text are not necessary in special cases that are relevant for us. The third concerns a simplification of dynamical coherence measures based on relative entropies.
	We begin with a Lemma concerning optimizations over separable states.
	\begin{lemma}\label{lem:obvious} 
		Let $D(\rho,\sigma)$ be a divergence and $\Theta, \Lambda$ CPTP. Then
		\begin{alignat}{1}
			\maxi{\sigma \in \mathcal{W}_{AE_A|BE_B}} &D\left(\left(\Theta^{AB}\Delta^{AB}\otimes\id^{E_A E_B}\right) \sigma, \left(\Lambda^{AB}\Delta^{AB}\otimes\id^{E_A E_B}\right) \sigma\right) \nonumber \\
			&=\maxi{\sigma \in \mathcal{W}_{A|B}} D\left(\Theta^{AB}\Delta^{AB} \sigma, \Lambda^{AB}\Delta^{AB} \sigma\right) 
		\end{alignat}
		and the same holds for the measured version, i.e.,
		\begin{alignat}{1}
			\supr{M \in \mathcal{M}}\ \maxi{\sigma \in \mathcal{W}_{AE_A|BE_B}} &D\left(M\left(\Theta^{AB}\Delta^{AB}\otimes\id^{E_A E_B}\right) \sigma, M \left(\Lambda^{AB}\Delta^{AB}\otimes\id^{E_A E_B}\right) \sigma\right) \nonumber \\
			=&\supr{M \in \mathcal{M}}\ \maxi{\sigma \in \mathcal{W}_{A|B}} D\left(M\Theta^{AB}\Delta^{AB} \sigma,M \Lambda^{AB}\Delta^{AB} \sigma\right),
		\end{alignat}
		where $\mathcal{M}$ can either be the set of free destructive measurements in $\mathcal{E}$ or the set of all destructive measurements.
	\end{lemma}
	\begin{proof}
		This Lemma is a direct consequence of the fact that the dephasing destroys quantum correlations. For their proof, we derive two inequalities each for the version with and without destructive measurements.
		
		a) Let $\tau$ and $\rho$ be quantum states. Using that $D$ is contractive by assumption, we find
		\begin{alignat}{1}
			\maxi{\sigma \in \mathcal{W}_{A|B}}& D\left(\Theta^{AB}\Delta^{AB} \sigma, \Lambda^{AB}\Delta^{AB} \sigma \right) \nonumber \\
			=&\maxi{\sigma \in \mathcal{W}_{A|B}} D\left(\tr_{E_AE_B}\left(\left(\Theta^{AB}\Delta^{AB}\otimes\id^{E_A E_B}\right)\left( \sigma^{AB}\otimes \tau^{E_A}\otimes \rho^{E_B}\right)\right),\right. \nonumber \\
			& \qquad \qquad \qquad \qquad \qquad\left. \tr_{E_AE_B}\left(\left(\Lambda^{AB}\Delta^{AB} \otimes\id^{E_A E_B}\right)\left(\sigma^{AB}\otimes \tau^{E_A}\otimes \rho^{E_B}\right)\right)\right) \nonumber \\
			\le&\maxi{\sigma \in \mathcal{W}_{A|B}} D\left(\left(\Theta^{AB}\Delta^{AB}\otimes\id^{E_A E_B}\right)\left( \sigma^{AB}\otimes \tau^{E_A}\otimes \rho^{E_B}\right), \left( \Lambda^{AB}\Delta^{AB} \otimes\id^{E_A E_B} \right) \left(\sigma^{AB}\otimes \tau^{E_A}\otimes \rho^{E_B}\right)\right) \nonumber \\
			\le&\maxi{\sigma \in \mathcal{W}_{AE_A|BE_B}} D\left(\left(\Theta^{AB}\Delta^{AB}\otimes\id^{E_A E_B}\right) \sigma, \left( \Lambda^{AB}\Delta^{AB}\otimes\id^{E_A E_B}\right) \sigma\right).
		\end{alignat}
		Denoting by $\tilde{M}\in\mathcal{M}$ a CPTP map associated to a destructive measurement with output independent of the input state, we find in complete analogy
		\begin{alignat}{1}
			\supr{M \in \mathcal{M}}\ \maxi{\sigma \in \mathcal{W}_{A|B}} &D\left(M\Theta^{AB}\Delta^{AB} \sigma, M\Lambda^{AB}\Delta^{AB} \sigma \right) \nonumber \\
			=&\supr{M \in \mathcal{M}}\ \maxi{\sigma \in \mathcal{W}_{A|B}}  D\left(\tr_{E_AE_B}\left(\left(M\Theta^{AB}\Delta^{AB}\otimes\tilde{M}^{E_AE_B}\right)\left( \sigma^{AB}\otimes \tau^{E_A}\otimes \rho^{E_B}\right)\right),\right. \nonumber \\
			&  \qquad \qquad \qquad \qquad \qquad \qquad\left. \tr_{E_AE_B}\left(\left(M\Lambda^{AB}\Delta^{AB} \otimes\tilde{M}^{E_AE_B}\right)\left(\sigma^{AB}\otimes \tau^{E_A}\otimes \rho^{E_B}\right)\right)\right) \nonumber \\
			\le&\supr{M \in \mathcal{M}}\ \maxi{\sigma \in \mathcal{W}_{A|B}}  D\left(\left(M\Theta^{AB}\Delta^{AB}\otimes\tilde{M}^{E_AE_B}\right)\left( \sigma^{AB}\otimes \tau^{E_A}\otimes \rho^{E_B}\right),\right. \nonumber \\
			&  \qquad \qquad \qquad\qquad \qquad \qquad\left. \left(M \Lambda^{AB}\Delta^{AB} \otimes\tilde{M}^{E_AE_B} \right) \left(\sigma^{AB}\otimes \tau^{E_A}\otimes \rho^{E_B}\right)\right) \nonumber \\
			\le&\supr{M \in \mathcal{M}}\ \maxi{\sigma \in \mathcal{W}_{AE_A|BE_B}}  D\left(M\left(\Theta^{AB}\Delta^{AB}\otimes\id^{E_A E_B}\right) \sigma, M\left( \Lambda^{AB}\Delta^{AB}\otimes\id^{E_A E_B}\right) \sigma\right).
		\end{alignat}
		b) We use Lem.~14 of Ref.~\cite{Theurer2019} to note that, if $\sigma \in \mathcal{W}_{AE_A|BE_B}$, then
		\begin{align}
			\left(\Delta^{AB}\otimes\id^{E_A E_B}\right) \sigma =&\left(\Delta^{AB}\otimes\id^{E_A E_B}\right) \sum_i r_i \rho_i^{AE_A}\otimes\tau_i^{B E_B} \nonumber \\
			=&\sum_{i,j,k,l,m} r_i q^i_{j|k}p_k^i\ketbra{k}{k}^A \otimes \ketbra{\phi^i_{j|k}}{\phi^i_{j|k}}^{E_A} \otimes\ketbra{m}{m}^B \otimes \tilde{q}^i_{l|m}\tilde{p}_m^i\ketbra{\xi^i_{l|m}}{\xi^i_{l|m}}^{E_B},
		\end{align}
		where $r_i, p_k^i, \tilde{p}_m^i, q^i_{j|k}, \tilde{q}^i_{l|m}$ represent (conditional) probabilities and $\phi^i_{j|k},\xi^i_{l|m}$ normalized quantum states.
		Using that $D$ is jointly convex and contractive, this allows us to conclude that
		\begin{alignat}{2}
			\maxi{\sigma \in \mathcal{W}_{AE_A|BE_B}} &D\left(\left(\Theta^{AB}\Delta^{AB}\otimes\id^{E_A E_B}\right) \sigma, \left(\Lambda^{AB}\Delta^{AB}\otimes\id^{E_A E_B}\right) \sigma\right) \nonumber \\
			=&\maxi{\sigma \in \mathcal{W}_{AE_A|BE_B}}  D\left(  \left(\Theta^{AB}\otimes\id^{E_A E_B}\right) \left( \Delta^{AB}\otimes\id^{E_A E_B}\right) \sigma, \left(\Lambda^{AB}\otimes\id^{E_A E_B}\right)\left( \Delta^{AB}\otimes\id^{E_A E_B} \right) \sigma\right) \nonumber \\
			\le&\maxi{\text{decompositions}} 	\sum_{i,j,k,l,m} r_i q^i_{j|k}p_k^i \tilde{q}^i_{l|m}\tilde{p}_m^i \nonumber \\
			&   D\left(\left(\Theta^{AB}\otimes\id^{E_A E_B}\right)\left( \ketbra{k}{k}^A \otimes \ketbra{\phi^i_{j|k}}{\phi^i_{j|k}}^{E_A}\otimes\ketbra{m}{m}^B \otimes  \ketbra{\xi^i_{l|m}}{\xi^i_{l|m}}^{E_B}\right), \right. \nonumber \\
			&  \qquad \qquad \qquad \qquad\qquad \qquad \left. \left(\Lambda^{AB}\otimes\id^{E_A E_B} \right)\left( \ketbra{k}{k}^A \otimes \ketbra{\phi^i_{j|k}}{\phi^i_{j|k}}^{E_A}\otimes\ketbra{m}{m}^B \otimes  \ketbra{\xi^i_{l|m}}{\xi^i_{l|m}}^{E_B}\right)\right) \nonumber \\
			\le&\maxi{\ket{\phi},\ket{\xi}, k, m}  D\left(\left(\Theta^{AB}\otimes\id^{E_A E_B}\right)\left(  \ketbra{k}{k}^A \otimes \ketbra{\phi}{\phi}^{E_A}\otimes\ketbra{m}{m}^B \otimes \ketbra{\xi}{\xi}^{E_B}\right),\right. \nonumber \\
			&  \qquad \qquad \qquad \qquad\qquad \qquad\left. \left(\Lambda^{AB}\otimes\id^{E_A E_B}\right)\left(  \ketbra{k}{k}^A \otimes \ketbra{\phi}{\phi}^{E_A}\otimes\ketbra{m}{m}^B \otimes \ketbra{\xi}{\xi}^{E_B}\right)\right) \nonumber \\
			=&\maxi{k,m}\  D\left(\Theta^{AB} \left( \ketbra{k}{k}^A \otimes \ketbra{m}{m}^B \right), \Lambda^{AB} \left(\ketbra{k}{k}^A \otimes\ketbra{m}{m}^B\right)\right) \nonumber \\
			=&\maxi{k,m}\  D\left(\Theta^{AB}\Delta^{AB} \left( \ketbra{k}{k}^A \otimes \ketbra{m}{m}^B \right), \Lambda^{AB}\Delta^{AB}\left( \ketbra{k}{k}^A \otimes \ketbra{m}{m}^B \right)\right) \nonumber \\
			\le&\maxi{\sigma \in \mathcal{W}_{A|B}}  \  D\left(\Theta^{AB}\Delta^{AB} \sigma, \Lambda^{AB}\Delta^{AB}\sigma\right).
		\end{alignat}
		
		For the case that includes destructive measurements, use in addition that for a POVM given by the elements $\Pi_i^{AB}$ and a state $\rho_A$, we can define a reduced POVM with elements $\tilde{\Pi}_i^B(\rho)$ acting only on system $B$ by 
		\begin{align}
			\tr\left(\Pi_i^{AB} \left(\rho^{A}\otimes\sigma^{B}\right)\right) = \tr\left(\Pi_i^{AB} \left(\rho^{A}\otimes\id^{B}\right)\left(\id^{A}\otimes\sigma^{B}\right)\right)= \tr\left(\tr_A\left(\Pi_i^{AB} \left(\rho^{A}\otimes\id^{B}\right)\right)\sigma^{B}\right)=:\tr\left(\tilde{\Pi}_i^B(\rho)\sigma^{B}\right),
		\end{align}
		i.e., $\tilde{\Pi}_i^B(\rho)=\tr_A\left(\Pi_i^{AB} \left(\rho^{A}\otimes\id^{B}\right)\right)$.
		With this in mind, we find
		\begin{alignat}{2}
			\supr{M \in \mathcal{M}}&\ \maxi{\sigma \in \mathcal{W}_{AE_A|BE_B}} D\left(M\left(\Theta^{AB}\Delta^{AB}\otimes\id^{E_A E_B}\right) \sigma, M\left(\Lambda^{AB}\Delta^{AB}\otimes\id^{E_A E_B}\right) \sigma\right) \nonumber \\
			=&\supr{M \in \mathcal{M}}\ \maxi{\sigma \in \mathcal{W}_{AE_A|BE_B}}  D\left( M \left(\Theta^{AB}\otimes\id^{E_A E_B}\right) \left( \Delta^{AB}\otimes\id^{E_A E_B}\right) \sigma, M\left(\Lambda^{AB}\otimes\id^{E_A E_B}\right)\left( \Delta^{AB}\otimes\id^{E_A E_B} \right) \sigma\right) \nonumber \\
			\le&\supr{M \in \mathcal{M}}\ \maxi{\text{decompositions}} 	\sum_{i,j,k,l,m} r_i q^i_{j|k}p_k^i \tilde{q}^i_{l|m}\tilde{p}_m^i \nonumber \\
			&   D\left(M\left(\Theta^{AB}\otimes\id^{E_A E_B}\right)\left( \ketbra{\phi^i_{j|k}}{\phi^i_{j|k}}^{E_A}\otimes\ketbra{k}{k}^A \otimes \ketbra{\xi^i_{l|m}}{\xi^i_{l|m}}^{E_B}\otimes\ketbra{m}{m}^B\right), \right. \nonumber \\
			& \qquad \qquad \qquad \qquad \qquad \qquad\left. M\left(\Lambda^{AB}\otimes\id^{E_A E_B} \right)\left( \ketbra{\phi^i_{j|k}}{\phi^i_{j|k}}^{E_A}\otimes\ketbra{k}{k}^A \otimes \ketbra{\xi^i_{l|m}}{\xi^i_{l|m}}^{E_B}\otimes\ketbra{m}{m}^B\right)\right) \nonumber \\
			\le&\supr{M \in \mathcal{M}}\ \maxi{\ket{\phi},\ket{\xi}, k, m}  D\left(M\left(\Theta^{AB}\otimes\id^{E_A E_B}\right)\left(  \ketbra{\phi}{\phi}^{E_A}\otimes\ketbra{k}{k}^A \otimes \ketbra{\xi}{\xi}^{E_B}\otimes\ketbra{m}{m}^B\right),\right. \nonumber \\
			& \qquad \qquad \qquad \qquad \qquad \qquad\left. M\left(\Lambda^{AB}\otimes\id^{E_A E_B}\right)\left( \ketbra{\phi}{\phi}^{E_A}\otimes\ketbra{k}{k}^A \otimes \ketbra{\xi}{\xi}^{E_B}\otimes\ketbra{m}{m}^B\right)\right) \nonumber \\
			=&\supr{\tilde{M}(\ket{\phi}\otimes\ket{\xi}) \in \mathcal{M}}\ \maxi{k,m}\  D\left(\tilde{M}(\ket{\phi}\otimes\ket{\xi})\Theta^{AB} \left( \ketbra{k}{k}^A \otimes \ketbra{m}{m}^B \right), \tilde{M}(\ket{\phi}\otimes\ket{\xi})\Lambda^{AB} \left(\ketbra{k}{k}^A \otimes\ketbra{m}{m}^B\right)\right) \nonumber \\
			\le&\supr{M \in \mathcal{M}}\ \maxi{k,m}\  D\left(M\Theta^{AB}\Delta^{AB} \left( \ketbra{k}{k}^A \otimes \ketbra{m}{m}^B \right), M\Lambda^{AB}\Delta^{AB}\left( \ketbra{k}{k}^A \otimes \ketbra{m}{m}^B \right)\right) \nonumber \\
			\le&\supr{M \in \mathcal{M}}\ \maxi{\sigma \in \mathcal{W}_{A|B}}  \  D\left(M\Theta^{AB}\Delta^{AB} \sigma, M\Lambda^{AB}\Delta^{AB}\sigma\right).
		\end{alignat}
		Having established inequalities in both directions for each version, the proof is finished.
	\end{proof}
	Next we prove a similar result concerning the case where we optimize over incoherent states.
	\begin{lemma}\label{lem:obvious2}
		Let $D(\rho, \sigma)$ be a divergence and $\Theta$, $\Lambda$ CPTP. Then
		\begin{align}
			\maxi{\sigma \in \mathcal{I}}\ D\left(\Theta \sigma,\Lambda \sigma \right)= \maxi{\sigma \in \mathcal{I}}\ D\left(\left(\Theta\otimes\id\right) \sigma,\left(\Lambda\otimes\id\right) \sigma \right).
		\end{align}
	\end{lemma}
	\begin{proof}
		This is a simple consequence of the fact that incoherent states on the joint system are separable. Using contractivity, we find
		\begin{align}
			\maxi{\sigma \in \mathcal{I}}\ D\left(\Theta \sigma,\Lambda \sigma \right)=& \maxi{\sigma \in \mathcal{I}}\ D\left(\tr_B\left(\left(\Theta^A\otimes\id^B\right) \sigma\right),\tr_B\left(\left(\Lambda^A\otimes\id^B\right) \sigma\right) \right)\nonumber \\
			\le&\maxi{\sigma \in \mathcal{I}}\ D\left(\left(\Theta\otimes\id\right) \sigma,\left(\Lambda\otimes\id\right) \sigma \right).
		\end{align}
		To show the reverse inequality, we use joint convexity, contractivity, and the fact that we can decompose the incoherent states on two systems as 
		\begin{align}
			\sigma = \sum_{ij} p_{ij} \ketbra{i}{i}\otimes \ketbra{j}{j},
		\end{align}
		where $p_{ij}$ denotes a probability distribution.  We then have
		\begin{align}
			\maxi{\sigma \in \mathcal{I}}\ D\left(\left(\Theta\otimes\id\right) \sigma,\left(\Lambda\otimes\id\right) \sigma \right) = &\maxi{p_{ij}}\ D\left(\left(\Theta\otimes\id\right) \sum_{ij} p_{ij} \ketbra{i}{i}\otimes \ketbra{j}{j},\left(\Lambda\otimes\id\right) \sum_{ij} p_{ij} \ketbra{i}{i}\otimes \ketbra{j}{j} \right) \nonumber \\
			\le& \maxi{p_{ij}}\ \sum_{ij} p_{ij} D\left(\Theta \ketbra{i}{i}\otimes \ketbra{j}{j},\Lambda \ketbra{i}{i}\otimes \ketbra{j}{j} \right) \nonumber \\
			\le& \maxi{i,j} \ D\left(\Theta \ketbra{i}{i}\otimes \ketbra{j}{j},\Lambda \ketbra{i}{i}\otimes \ketbra{j}{j} \right) \nonumber \\
			=&\maxi{i} \ D\left(\Theta \ketbra{i}{i},\Lambda \ketbra{i}{i} \right) \nonumber \\
			\le&\maxi{\sigma \in \mathcal{I}}\ D\left(\Theta \sigma,\Lambda \sigma \right),
		\end{align}
		which finishes the proof.
	\end{proof}

	To conclude this section, we present a Lemma which simplifies certain coherence measures and had been shown previously in Ref.~\cite{Liu2020} for the case of MIO.

	\begin{lemma}\label{lem:semiClosedExt}
		For the relative entropy $S(\rho,\sigma)$,
		\begin{align}\label{eq:Lemma}
			C_{\mathcal{C},S}^{\mathcal{I},{\rm  no}}\left(\Theta\right)=\maxi{\sigma \in \mathcal{I}}\ \left[ S(\Delta \Theta \sigma)-S(\Theta \sigma) \right] \nonumber
		\end{align}	
		holds, where $S(\rho)$ denotes the von Neumann entropy.
	\end{lemma}
	This shows that $C_{\mathcal{C},S}^{\mathcal{I},{\rm  no}}\left(\Theta\right)$ is a coherence power, since
	\begin{align}
		C_{\mathcal{C},S}^{\mathcal{I},{\rm  no}}\left(\Theta\right)=&\maxi{\sigma \in \mathcal{I}}\ R_C(\Theta\sigma),
	\end{align}
	where $R_C$ denotes the relative entropy of coherence~\cite{Aberg2006,Baumgratz2014,Winter2016}. In addition, the Lemma shows that these measures can only be faithful for MIO: the right hand side of Eq.~(\ref{eq:Lemma}) is independent of the choice of $\mathcal{C}$ and, by construction, the measure with respect to MIO is zero on all operations in MIO. Since, e.g., IO is a strict subset of MIO, there exist operations outside IO on which the respective measure is zero. 
	
	\begin{proof}
		We begin by reminding the well known fact that for quantum states $\rho$ and $\sigma$, we have
		\begin{align}
			S\left(\rho,\Delta\sigma\right)=S(\Delta \rho)-S(\rho)+S\left(\Delta \rho, \Delta \sigma\right).
		\end{align}
		Together with Lem.~\ref{lem:obvious2} follows
		\begin{alignat}{2}
			C_{\mathcal{C},S}^{\mathcal{I},{\rm no}}\left(\Theta\right)=&\mini{\Lambda \in \mathcal{C}}\ \maxi{\sigma \in \mathcal{I}_{AB}} &&S\left(\left(\Theta^A\otimes \id^B\right)\sigma,\left(\Lambda^A\otimes \id^B\right)\sigma \right) \nonumber\\
			=&\mini{\Lambda \in \mathcal{C}}\ \maxi{\sigma \in \mathcal{I}}\ &&S\left(\Theta\sigma,\Lambda\sigma \right) \nonumber\\
			=&\mini{\Lambda \in \mathcal{C}}\ \maxi{\sigma \in \mathcal{I}}\ &&\left[S(\Delta \Theta \sigma)-S(\Theta \sigma)+S\left(\Delta \Theta \sigma , \Lambda \sigma\right)\right].
		\end{alignat}
		Next we remind that the quantum operation $\Lambda=\Delta \Theta \Delta$ can be represented by Kraus operators $L_{ijn}=\ketbra{j}{j}K_n\ketbra{i}{i}$ if $\Theta$ is given by Kraus operators $K_n$. Therefore such a $\Lambda$ is contained in SIO, which can be characterized by every Kraus operator having at most one non-zero entry per column and row~\cite{Winter2016}. Since SIO is included in LOP, DIO, IO, and MIO, $\Lambda$ is also in all of these sets. 
		Together with the fact that the relative entropy between two quantum states is non-negative and that $S\left(\Delta \Theta \sigma , \Delta \Theta \Delta \sigma\right)=0$, we arrive at our statement.
	\end{proof}

	\section{Proofs of the results in the main text} \label{sec:proofs}
	\setcounter{theorem}{1}
	Here, with the help of the technical Lemmas from Sec.~\ref{sec:lemmas}, we prove the results presented in the main text, which we restate for readability. In addition to the sets of free operations discussed in the main text, we will also consider the subsets of $\rm LOCC$ that can be implemented with $\rm r$ rounds of classical communication ($\rm LOCC_r$)~\cite{Bennett1999}. Since $\rm LOCC_r$ is not closed under concatenation, i.e., for $\Phi_i\in {\rm LOCC_2}$ it is possible that $\Phi_2 \Phi_1 \in {\rm LOCC_4}$, it seems difficult to include $\rm LOCC_r$ into our framework directly: as we will see in the following, to prove that our monotones defined in the main text are indeed monotonic under concatenation with free operations, we rely on the fact that both $\mathcal{C}$ and $\mathcal{E}$ are closed under concatenation. In addition, this implies that, according to the definitions we chose, there cannot be a monotone with respect to $\rm LOCC_r$  which is also faithful. A method to circumvent this problem could be to define monotonicity not with respect to concatenations with free operations, but with respect to so called maximally free superoperations, i.e., linear maps from quantum operations to quantum operations that map $\rm LOCC_r$ to itself. However, it is not clear how powerful this set of maximally free superoperations would be and it does not have an equally strong physical motivation in terms of circuit quantum computation. 
	
	Nevertheless, in our proofs in Sec.~\ref{sec:reduction}, the quantities
	\begin{align} \label{eq:LOCCrMeas}
		&E_{{\rm LOCC_r}^{A|B},D}^{\mathcal{S},\mathcal{M}}\left(\Theta\right):=\infi{\Lambda\in \rm {LOCC_r}^{A|B}}\ D^{\mathcal{S},\mathcal{M}}(\Theta, \Lambda)
	\end{align}
	will be of use, since, as we discuss in Sec.~\ref{sec:relMeasures}, they bound the monotones introduced in the main text and in contrast to $\rm LOCC$, its subsets $\rm LOCC_r$ are closed. Therefore, we will prove some of our results from the main text also for $\rm LOCC_r$.

	\begin{proposition}
		Let $\mathcal{S}$ denote either the set of all quantum states or the set of free quantum states, $\mathcal{M} \in  \left\{ \rm{all,free, no} \right\}$, $\tilde{\mathcal{M} }\in  \left\{ \rm{all, no} \right\}$, $\tilde{\mathcal{E}} \in \left\{\rm{\overline{LOCC}}, \rm{SEP} \right\}$, and $\tilde{\mathcal{C}} \in \{\rm LOP, IO ,MIO \}$.
		Then 
		\begin{align*}
			E_{\tilde{\mathcal{E} }^{A|B},D}^{\mathcal{S},\mathcal{M}}\left(\Theta\right),\ C_{{\rm MIO},D}^{\mathcal{S},\mathcal{M}}\left(\Theta\right),\ 
			C_{{\mathcal{C}},D}^{\mathcal{D},\tilde{\mathcal{M}}}\left(\Theta\right),\ \text{and}\ C_{{\tilde{\mathcal{C} }},D}^{\mathcal{D},{\rm free}}\left(\Theta\right)
		\end{align*}
		are convex resource measures. The remaining functionals from Def.~1 are convex resource monotones, with those defined via destructive DIO/SIO measurements vanishing on all operations.
	\end{proposition}
	
	\begin{proof}
		a) monotonicity: We begin by proving monotonicity for $E_{\mathcal{E}^{A|B},D}^{\mathcal{S},\rm{no}}\left(\Theta\right)$. Note that this has been shown in Refs.~\cite{Liu2020,Liu2019,Gour2019a} before, but for completeness, we repeat the proof here. 	
		Let $\Phi_1^{AE_A|BE_B},\Phi_2^{AE_A|BE_B} \in \mathcal{E}^{AE_A|BE_B}$. Then we have the following chain of inequalities, which are explained below:
		\begin{subequations}
			\begin{align}
				E&_{\mathcal{E}^{AE_A|BE_B},D}^{\mathcal{S},\rm{no}}\left(\Phi_2^{AE_A|BE_B} \left(\Theta^{AB} \otimes \id^{E_A E_B}\right)\Phi_1^{AE_A|BE_B}\right)\\
				&=\infi{\Lambda \in \mathcal{E}^{AE_A|BE_B}}\ \maxi{\sigma \in \mathcal{S}}\ D \left(\left( \Phi_2^{AE_A|BE_B} \left(\Theta^{AB} \otimes \id^{E_A E_B}\right)\Phi_1^{AE_A|BE_B}\otimes\id^{\tilde{A}\tilde{B}}\right)\sigma, \left(\Lambda^{AE_A|BE_B}\otimes\id^{\tilde{A} \tilde{B}}\right)\sigma \right) \label{eq:e1} \\
				&\le \infi{\Lambda \in \mathcal{E}^{A|B}}\ \maxi{\sigma \in \mathcal{S}}\ D \left(\left( \Phi_2^{AE_A|BE_B} \left(\Theta^{AB} \otimes \id^{E_A E_B}\right)\Phi_1^{AE_A|BE_B}\otimes\id^{\tilde{A}\tilde{B}}\right)\sigma, \right. \nonumber \\
				&\left.\qquad \qquad \qquad\qquad \qquad \qquad \qquad \qquad\quad \left( \Phi_2^{AE_A|BE_B} \left(\Lambda^{A|B} \otimes \id^{E_A E_B}\right)\Phi_1^{AE_A|BE_B}\otimes\id^{\tilde{A}\tilde{B}}\right)\sigma \right) \label{eq:e2}\\
				&\le \infi{\Lambda \in \mathcal{E}^{A|B}}\ \maxi{\sigma \in \mathcal{S}}\ D \left(\left( \Phi_2^{AE_A|BE_B} \left(\Theta^{AB} \otimes \id^{E_A E_B}\right)\otimes\id^{\tilde{A}\tilde{B}}\right)\sigma, \left( \Phi_2^{AE_A|BE_B} \left(\Lambda^{A|B} \otimes \id^{E_A E_B}\right)\otimes\id^{\tilde{A}\tilde{B}}\right)\sigma \right) \label{eq:e3}\\
				&\le \infi{\Lambda \in \mathcal{E}^{A|B}}\ \maxi{\sigma \in \mathcal{S}}\ D \left(  \left(\Theta^{AB} \otimes \id^{E_A E_B \tilde{A} \tilde{B}}\right)\sigma,  \left(\Lambda^{A|B}\otimes\id^{E_A E_B \tilde{A} \tilde{B}}\right)\sigma \right) \label{eq:e4} \\
				&=E_{\mathcal{E}^{A|B},D}^{\mathcal{S},\rm{no}}\left(\Theta^{AB}\right). \label{eq:e5}
			\end{align}
		\end{subequations}
		In (\ref{eq:e1}), we repeated the definition and in (\ref{eq:e2}), we used the fact that, by replacing the second argument, we take the infimum over a smaller set: by assumption $\Phi_2\left(\Lambda\otimes\id\right)\Phi_1$ is a free operation, since the set of free operations is closed under concatenation.  Line (\ref{eq:e3}) follows from an increase of the set of states over which we maximize, since $\Phi_1$ maps free states to free states. In (\ref{eq:e4}), we used contractivity of $D$ and, merging the spaces $E_A$ and $\tilde{A}$ as well as $E_B$ and $\tilde{B}$, (\ref{eq:e5}) follows again from the definition.

		For $E_{\mathcal{E}^{A|B},D}^{\mathcal{S},\mathcal{M}}\left(\Theta^{AB}\right)$ with $\mathcal{M} \in  \left\{ \rm{all,free} \right\}$ we find in complete analogy
		\begin{subequations}
			\begin{align}
				E_{\mathcal{E}^{AE_A|BE_B},D}^{\mathcal{S},\mathcal{M}}&\left(\Phi_2^{AE_A|BE_B} \left(\Theta^{AB} \otimes \id^{E_A E_B}\right)\Phi_1^{AE_A|BE_B}\right)\\
				=&\infi{\Lambda \in \mathcal{E}^{AE_A|BE_B}}\ \supr{M \in \mathcal{M}}\ \maxi{\sigma \in \mathcal{S}}\ D \left(M\left( \Phi_2^{AE_A|BE_B} \left(\Theta^{AB} \otimes \id^{E_A E_B}\right)\Phi_1^{AE_A|BE_B}\otimes\id^{\tilde{A}\tilde{B}}\right)\sigma, \right.\nonumber \\ 
				&\left.\qquad \qquad \qquad \qquad \qquad \qquad \qquad M\left(\Lambda^{AE_A|BE_B}\otimes\id^{\tilde{A} \tilde{B}}\right)\sigma \right) \label{eq:ee1} \\
				\le& \infi{\Lambda \in \mathcal{E}^{A|B}}\ \supr{M \in \mathcal{M}}\ \maxi{\sigma \in \mathcal{S}}\ D \left(M\left( \Phi_2^{AE_A|BE_B} \left(\Theta^{AB} \otimes \id^{E_A E_B}\right)\Phi_1^{AE_A|BE_B}\otimes\id^{\tilde{A}\tilde{B}}\right)\sigma, \right. \nonumber \\
				&\left.\qquad \qquad \qquad \qquad \qquad \qquad \qquad M\left( \Phi_2^{AE_A|BE_B} \left(\Lambda^{A|B} \otimes \id^{E_A E_B}\right)\Phi_1^{AE_A|BE_B}\otimes\id^{\tilde{A}\tilde{B}}\right)\sigma \right) \label{eq:ee2}\\
				\le& \infi{\Lambda \in \mathcal{E}^{A|B}}\ \supr{M \in \mathcal{M}}\ \maxi{\sigma \in \mathcal{S}}\ D \left(M  \left(\Theta^{AB} \otimes \id^{E_A E_B \tilde{A} \tilde{B}}\right)\sigma,  M\left(\Lambda^{A|B}\otimes\id^{E_A E_B \tilde{A} \tilde{B}}\right)\sigma \right) \label{eq:ee4} \\
				=&E_{\mathcal{E}^{A|B},D}^{\mathcal{S},\mathcal{M}}\left(\Theta^{AB}\right), \label{eq:ee5}
			\end{align}
		\end{subequations}
		where we used in addition in (\ref{eq:ee4}) that the (free) destructive measurements which are preceded by the free operation $\Phi_2$ form a subset of the (free) destructive measurements.
		The proofs for the monotones with respect to $\mathcal{C}$ use exactly the same arguments, which is why we will not repeat them here.
		
		b) faithfulness (of the functionals we claim to be measures): 
		For the proofs of faithfulness, let us begin with some general remarks. With the exception of LOCC, both the sets of free operations as well as the sets of free destructive measurements are closed and the infimum and supremum are achieved (therefore we can replace them with a minimum and a maximum). In addition, if $\Theta$ is in the set of free operations under consideration, we can choose $\Lambda=\Theta$ and the respective monotones are obviously zero. To prove the other direction, we consider different cases separately.

		First we consider the monotones in which we are maximizing over all states. If $\Theta$ is not free, for every fixed free $\Lambda$, there exist at least one state that transforms differently under the two operations (otherwise they would be the same, which is a contradiction). Using our assumptions on $D$, this proves faithfulness in the cases where we do not optimize over destructive measurements. In case we optimize over sets of informationally complete destructive measurements, there will necessarily exist an allowed destructive measurement that leads to different statistics for two different states. Next we recall that LOP contains all POVMs (and is a subset of IO and MIO)~\cite{Egloff2018}. Since all destructive measurements and all the sets of restricted destructive measurements in $\mathcal{E}$ are informationally complete~\cite{Caves2002,Piani2009} (see Refs.~\cite{Matthews2009,Lami2018} for some explicit bounds on the restricted trace distances), this proves again faithfulness for the respective monotones. 
		
		We now proceed to the monotones which include a maximization over free states and consider first the case that the free operations are either MIO or SEP. Since these form the maximal sets of operations that do transform free states to free states, if $\Theta$ is not free, there exists at least one free state such that $\Theta$ converts it to a non-free one. This is impossible by every free operation $\Lambda$ and to prove that the respective monotones where we maximize only over free states are faithful, we can follow up with the same arguments we used above.
		
		The remaining case is when we optimize over separable states and consider ${\rm \overline{LOCC}}$ as free operations. In addition to the arguments above, we now have to show that operations outside ${\rm \overline{LOCC}}$ cannot be simulated \textit{on separable states} by operations inside ${\rm \overline{LOCC}}$, which we will do now. The states which are separable with respect to the division into $AE_A$ and $BE_B$ over which we optimize contain the maximally entangled state $\ket{\Phi_+^{A B| E_A E_B}}=\frac{1}{\sqrt{d_A d_B}}\sum_{ij}\ket{ii}^{AE_A} \otimes\ket{jj}^{BE_B}$ \textit{between $AB$ and $E_AE_B$}. Using the above arguments, it therefore follows that, if the monotone is zero, then there exists a free $\Lambda$ such that
		\begin{align}
			\left(\Theta^{AB}\otimes\id^{E_A E_B}\right) \ketbra{\Phi_+^{A B| E_A E_B}}{\Phi_+^{A B| E_A E_B}} = \left(\Lambda^{A|B}\otimes\id^{E_A E_B }\right) \ketbra{\Phi_+^{A B| E_A E_B}}{\Phi_+^{A B| E_A E_B}}, 
		\end{align}
		which, using the Choi isomorphism, is equivalent to saying that $\Theta=\Lambda$. Thus the monotones are faithful. This argument was also used in Ref.~\cite{Gour2019a}. 	
		At this point, we remark that the same arguments allow to show that $E_{{\rm LOCC_r}^{A|B},D}^{\mathcal{S},\mathcal{M}}\left(\Theta\right)$ as defined in Eq.~(\ref{eq:LOCCrMeas}) is faithful, i.e., it is zero if and only if $\Theta \in {\rm LOCC_r}$. 
		
		We also note that measurements which are in DIO cannot detect coherence in the sense that the measurement statistics are determined by the populations alone. Since SIO and DIO can be used to implement arbitrary transformations on the populations, the monotones where we only allow for destructive SIO/DIO measurements are always zero, i.e., useless.
		
		c) convexity: This is an immediate consequence of joint convexity of $D$ and the convexity of the free operations. We will show the proof for the example of $E_{\mathcal{E}^{A|B},D}^{\mathcal{S},\mathcal{M}}\left(\Theta^{AB}\right)$ with $\mathcal{M}\ne \rm no$. For all other monotones, the proofs are exactly analogous, which is why we will not repeat them here. 
		For $0\le t\le1$, we find
		\begin{align}
			E_{\mathcal{E}^{A|B},D}^{\mathcal{S},\mathcal{M}}&\left(t \Theta_1^{AB}+(1-t) \Theta_2^{AB}\right) \nonumber\\
			=&\infi{\Lambda \in \mathcal{E}^{A|B}}\ \supr{M \in \mathcal{M}}\ \maxi{\sigma \in \mathcal{S}}\ D \left(M\left(\left(t \Theta_1^{AB}+(1-t) \Theta_2^{AB}\right)\otimes\id^{E_A E_B}\right)\sigma, M\left(\Lambda^{A|B}\otimes\id^{E_A E_B }\right)\sigma \right)  \nonumber\\
			\le&\infi{\Lambda_i \in \mathcal{E}^{A|B}}\ \supr{M \in \mathcal{M}}\ \maxi{\sigma \in \mathcal{S}}\ D \left(M\left(\left(t \Theta_1^{AB}+(1-t) \Theta_2^{AB}\right)\otimes\id^{E_A E_B}\right)\sigma, M\left(\left(t \Lambda_1^{A|B}+(1-t) \Lambda_2^{A|B}\right)\otimes\id^{E_A E_B }\right)\sigma \right)  \nonumber\\
			\le&\infi{\Lambda_i \in \mathcal{E}^{A|B}}\ \supr{M \in \mathcal{M}}\ \maxi{\sigma \in \mathcal{S}} \left[t D \left(M \left( \Theta_1^{AB}\otimes\id^{E_A E_B}\right)\sigma, M\left( \Lambda_1^{A|B}\otimes\id^{E_A E_B }\right)\sigma \right)  \right.  \nonumber\\ 
			&\left.\qquad \qquad \qquad \qquad \qquad+(1-t) D \left(M \left( \Theta_2^{AB}\otimes\id^{E_A E_B}\right)\sigma, M\left( \Lambda_2^{A|B}\otimes\id^{E_A E_B }\right)\sigma \right)    \right]  \nonumber\\
			\le&t\ \infi{\Lambda_1 \in \mathcal{E}^{A|B}}\ \supr{M \in \mathcal{M}}\ \maxi{\sigma \in \mathcal{S}} \left[ D \left(M \left( \Theta_1^{AB}\otimes\id^{E_A E_B}\right)\sigma, M\left( \Lambda_1^{A|B}\otimes\id^{E_A E_B }\right)\sigma \right)  \right]  \nonumber\\ 
			&\qquad \qquad \qquad \qquad \qquad +(1-t)\ \infi{\Lambda_2 \in \mathcal{E}^{A|B}}\  \supr{M \in \mathcal{M}}\ \maxi{\sigma \in \mathcal{S}} \left[  D \left(M \left( \Theta_2^{AB}\otimes\id^{E_A E_B}\right)\sigma, M\left( \Lambda_2^{A|B}\otimes\id^{E_A E_B }\right)\sigma \right)    \right]  \nonumber\\
			=&t \ E_{\mathcal{E}^{A|B},D}^{\mathcal{S},\mathcal{M}}\left(\Theta^{AB}_1\right) + (1-t) \ E_{\mathcal{E}^{A|B},D}^{\mathcal{S},\mathcal{M}}\left(\Theta^{AB}_2\right),
		\end{align}
		which finishes the proof of convexity.
	\end{proof}
	
	\begin{theorem}
		Let $\Phi_i\in \mathcal{C}$. Then	
		\begin{align}
			&C_{{\mathcal{C}},D}^{\mathcal{I},{\rm no}}\left(\Theta\right)   \ge E_{{\mathcal{E}}^{A|B},D}^{\mathcal{W},{\rm no}}\left(\Phi_2\left(\Theta^A\otimes \id^B\right)\Phi_1\Delta\right) \nonumber
		\end{align}
		and
		\begin{align}
			&C_{{\mathcal{C}},D}^{\mathcal{I},{\rm all}}\left(\Theta\right)   \ge E_{{\mathcal{E}}^{A|B},D}^{\mathcal{W},{\rm all}}\left(\Phi_2\left(\Theta^A\otimes \id^B\right)\Phi_1\Delta\right). \nonumber
		\end{align}
	\end{theorem}
	Note: As will be apparent from the proof, the Theorem also holds true for $\mathcal{E}={\rm LOCC_{r\ge3}}$ as defined in Eq.~(\ref{eq:LOCCrMeas}).
	\begin{proof}
		Using Lem.~\ref{lem:obvious} and our assumptions on $D$, we find
		\begin{subequations}
			\begin{align}
				C_{{\rm MIO},D}^{\mathcal{I},{\rm no}}\left(\Theta\right)=&\mini{\Lambda \in {\rm MIO} }\ \maxi{\sigma \in \mathcal{I}_{AB}}  D\left(\left(\Theta^A\otimes \id^B\right)\sigma,\left(\Lambda^A\otimes \id^B\right)\sigma \right) \\
				\ge& \mini{\Lambda \in {\rm MIO} }\ \maxi{\sigma \in \mathcal{I}_{AB}}  D\left(\Phi^{AB}_2\left(\Theta^A\otimes \id^B\right)\Phi_1^{AB} \sigma, \Phi^{AB}_2\left(\Lambda^A\otimes \id^B\right)\Phi_1^{AB} \sigma   \right)  \\
				=& \mini{\Lambda \in {\rm MIO} }\ \maxi{\sigma \in \mathcal{I}_{AB}}  D\left(\Phi^{AB}_2\left(\Theta^A\otimes \id^B\right)\Phi_1^{AB} \Delta^{AB} \sigma, \Phi^{AB}_2\left(\Lambda^A\otimes \id^B\right)\Phi_1^{AB} \Delta^{AB} \sigma   \right) \\
				=& \mini{\Lambda \in {\rm MIO} }\ \maxi{\sigma \in \mathcal{W}_{A|B}}  D\left(\Phi^{AB}_2\left(\Theta^A\otimes \id^B\right)\Phi_1^{AB} \Delta^{AB} \sigma,  \Phi^{AB}_2\left(\Lambda^A\otimes \id^B\right)\Phi_1^{AB}  \Delta^{AB} \sigma   \right)  \label{eq:expl}\\
				=&\mini{\Lambda \in {\rm MIO} }\  \maxi{\sigma \in \mathcal{W}_{AE_A|BE_B}}  D\left(\left(\Phi^{AB}_2\left(\Theta^A\otimes \id^B\right)\Phi_1^{AB} \Delta^{AB}\otimes\id^{E_A E_B}\right) \sigma, \right. \nonumber \\
				&  \qquad \qquad \qquad \qquad\qquad \qquad \qquad \qquad \qquad \qquad \left. \left(\Phi^{AB}_2\left(\Lambda^A\otimes \id^B\right)\Phi_1^{AB} \Delta^{AB}\otimes\id^{E_A E_B} \right)\sigma\right)  \\
				\ge&\mini{\Lambda \in {\rm LOCC_3}^{A|B}}\  \maxi{\sigma \in \mathcal{W}_{AE_A|BE_B}}  D\left(\left(\Phi^{AB}_2\left(\Theta^A\otimes \id^B\right)\Phi_1^{AB} \Delta^{AB}\otimes\id^{E_A E_B}\right) \sigma, \left(\Lambda^{A|B}\otimes\id^{E_A E_B}\right) \sigma\right),
			\end{align}
		\end{subequations}
		where we used in Eq.~(\ref{eq:expl}) that $\Delta \sigma=\Delta \Delta \sigma$ and in the last line that 
		\begin{align}\label{eq:oneRound}
			\Phi^{AB}_2\left(\Lambda^A\otimes \id^B\right)\Phi_1^{AB} \Delta^{AB}
		\end{align}
		can be implemented using local operations and three rounds of classical communication: since $\Phi_i$ and $\Lambda$ are maximally incoherent, the states
		\begin{align}
			\sigma_{ij}:=\sum_{kl}p_{kl}^{ij} \ketbra{k}{k}\otimes\ketbra{l}{l}:=\Phi^{AB}_2\left(\Lambda^A\otimes \id^B\right)\Phi_1^{AB} \ketbra{i}{i}\otimes\ketbra{j}{j}
		\end{align}
		are separable. Therefore, the operation in Eq.~(\ref{eq:oneRound}) can be implemented by performing local projective measurements in the incoherent bases, sharing the outcomes, and preparing the corresponding states $\sigma_{ij}$. To do this, one might need shared randomness, which can be established whilst sharing the measurement outcomes. In detail, first Alice performs her projective measurement and shares the outcome as well as some randomness with Bob. This constitutes the first round of the protocol. In the second round, Bob then does his projective measurement, creates his local state (conditioned on the measurement outcomes and the randomness), and communicates his outcome to Alice. In the third round, Alice prepares  her local state, again conditioned on the measurement outcomes and the randomness.
		Together with the discussion in Sec.~\ref{sec:relMeasures}, this finishes the proof of the first part of our statement (and also proves the note we added above this proof). The second part follows in complete analogy, making again use of Lem.~\ref{lem:obvious} and the fact that $M \Phi_2$ is again a destructive measurement. 
	\end{proof}

	\begin{theorem}
		For $S(\rho,\sigma)$ the relative entropy and $\mathrm{dim} B=\mathrm{dim} A$,
		\begin{align}
			C_{{\mathcal{C}},S}^{\mathcal{I},{\rm no}}\left(\Theta\right)   \le E_{{\mathcal{E}}^{A|B},S}^{\mathcal{W},{\rm no}}\left(\mathcal{U}_\mathrm{CNOT}\left(\Theta^A\otimes \id^B\right)\Delta\right) \nonumber
		\end{align}
		holds.
	\end{theorem} 
	\begin{proof}
		Let us denote the relative entropy of entanglement~\cite{Vedral1998} with respect to the bi-partition into parties $A$ and $B$ by $R_E^{A|B}$, i.e.,
		\begin{align}
			R_E^{A|B}(\rho)=\mini{\sigma \in \mathcal{W}_{A|B}}S\left(\rho,\sigma\right).
		\end{align}
		For readability we also define 
		\begin{align}
			\Theta \sigma =: \rho_\sigma=\sum_{i,j} \rho_{ij}^\sigma\ketbra{i}{j}.
		\end{align}
		Then, with Lem.~\ref{lem:obvious}, and applying a technique used in Ref.~\cite{Streltsov2015}, we find
		\begin{subequations}
			\begin{align}
				E_{{\rm SEP}^{A|B},S}^{\mathcal{W},{\rm no}}\left(\mathcal{U}_\mathrm{CNOT}^{AB}\left(\Theta^A\otimes \id^B\right)\Delta^{AB}\right)
				=& \mini{\Lambda \in {\rm SEP}^{A|B}}\  \maxi{\sigma \in \mathcal{W}_{A|B}}  S\left(\mathcal{U}_\mathrm{CNOT}^{AB}\left(\Theta^A\otimes \id^B\right)\Delta^{AB} \sigma, \Lambda^{A|B} \sigma   \right) \label{eq:inv1}  \\
				\ge&  \maxi{\sigma \in \mathcal{W}_{A|B}} \ \mini{\Lambda \in {\rm SEP}^{A|B}}  S\left(\mathcal{U}_\mathrm{CNOT}^{AB}\left(\Theta^A\otimes \id^B\right)\Delta^{AB} \sigma, \Lambda^{A|B} \sigma   \right)  \label{eq:inv2} \\
				=&\maxi{\sigma \in \mathcal{W}_{A|B}}\  R_E^{A|B}\left(\mathcal{U}_\mathrm{CNOT}^{AB}\left(\Theta^A\otimes \id^B\right)\Delta^{AB} \sigma  \right) \label{eq:inv3}  \\
				\ge& \maxi{\sigma \in \mathcal{I}_A}\    R_E^{A|B}\left(\mathcal{U}_\mathrm{CNOT}^{AB}\left(\Theta^A\otimes \id^B\right)\Delta^{AB} \left(\sigma^A \otimes \ketbra{0}{0}^B\right)  \right) \label{eq:inv4}  \\
				=&  \maxi{\sigma \in \mathcal{I}_A}\    R_E^{A|B}\left(\sum_{i,j} \rho_{ij}^\sigma\ketbra{i}{j}^A\otimes  \ketbra{i}{j}^B\right) \label{eq:inv5}  \\
				\ge& \maxi{\sigma \in \mathcal{I}_A} \  S\left(\Delta \Theta \sigma\right)-S\left(\Theta \sigma\right), \label{eq:inv6}
			\end{align}
		\end{subequations}
		where we used the max-min inequality in Eq.~(\ref{eq:inv2}) and in Eq.~(\ref{eq:inv3}), we used that with SEP, we can prepare an arbitrary separable state, but no state outside the set of separable states if we have a separable input state.
		In the last line, we used that~\cite{Plenio2000}
		\begin{align}
			R_E^{A|B}\left(\rho^{AB}\right)\ge S\left(\rho^A\right)-S\left(\rho^{AB}\right).
		\end{align}
		Together with Lem.~\ref{lem:semiClosedExt} and the discussion in Sec.~\ref{sec:relMeasures}, this finishes the proof. Note that all our arguments also hold if we take $\mathcal{E}={ \rm LOCC_{r}}$. Thus the Theorem also holds in this case.
	\end{proof}
	
	\begin{theorem}
		For $S(\rho,\sigma)$ the relative entropy,
		\begin{align}
			C_{{\mathcal{C}},S}^{\mathcal{I},{\rm no}}\left(\Theta\right)   = \supr{\Phi_1,\Phi_2 \in \mathcal{C}}\  E_{{\mathcal{E}}^{A|B},S}^{\mathcal{W},{\rm no}}\left(\Phi_2\left(\Theta^A\otimes \id^B\right)\Phi_1\Delta\right). \nonumber
		\end{align}
		The supremum is achieved for $\mathrm{dim} B=\mathrm{dim} A$, $\Phi_1=\id$ and $\Phi_2=\mathcal{U}_\mathrm{CNOT}$. 
	\end{theorem}
	\begin{proof}
		The proof follows directly from Thms.~\ref{thm:inequality} and~\ref{thm:inverse}. As we noted in the respective proofs, these two Theorems also hold if we choose $\mathcal{E}= {\rm LOCC_{r\ge3}}$, which is why also this Theorem holds for $\mathcal{E}= {\rm LOCC_{r\ge3}}$.
	\end{proof}
	
	\begin{corollary}
		An operation $\Theta$ can be converted to an operation outside $\mathcal{E}$ with operations in $\mathcal{C}$ if and only if $\Theta$ is not in $\rm MIO$. 
	\end{corollary}
	\begin{proof}
		According to Prop.~\ref{prop:monotones} and its proof, the two measures 
		\begin{align}
			C_{{\rm MIO},S}^{\mathcal{I},{\rm no}},  E_{{\rm SEP}^{A|B},S}^{\mathcal{W},{\rm no}} \nonumber
		\end{align}	
		as well as
		\begin{align*}
			E_{{\rm LOCC_{r\ge3}}^{A|B},S}^{\mathcal{W},{\rm no}}
		\end{align*}
		are faithful, i.e., they are zero if and only if they are evaluated on free operations. In addition, the inequalities discussed in Sec.~\ref{sec:relMeasures} hold. Together with Thm.~\ref{thm:main} (and in particular that the optimal conversion scheme is contained in SIO) this finishes the proof and shows that the Corollary also holds for $\mathcal{E}={\rm LOCC_{r\ge3}}$.
	\end{proof}
	
	\begin{theorem}
		Let $E_{\mathcal{E}^{AE_A|B}}$ be a (convex) resource monotone with respect to the set of operations $\mathcal{E}^{AE_A|B}$. Then 
		\begin{align}\label{eq:defCohWithEntSM}
			C_{\mathcal{C}}^{\mathcal{E},E}\left(\Theta\right):=\supr{\Lambda_i\in \mathcal{C}}\ E_{\mathcal{E}^{AE_A|B}}\left(\Lambda_2\left(\Theta^{A} \otimes \id^{E_AB} \right)\Lambda_1\Delta\right) 
		\end{align}
		is a (convex) resource monotone with respect to the operations $\mathcal{C}$. If $E_{\mathcal{E}^{AE_A|B}}$ is in addition faithful, Eq.~(\ref{eq:defCohWithEntSM}) defines a measure with respect to $\rm MIO$.
	\end{theorem}
	\begin{proof}
		Non-negativity is obviously inherited and monotonicity is a simple consequence of the fact that $\mathcal{C}$ is closed under concatenation. Let $\Phi_i \in \mathcal{C}$. Then we have
		\begin{align}
			C_{\mathcal{C}}^{\mathcal{E},E} \left( \Phi_2 \left(\Theta^A \otimes \id^{\tilde{A}} \right) \Phi_1 \right) =& \supr{\Lambda_i\in \mathcal{C}}\ E_{\mathcal{E}^{A\tilde{A}E_A|B}}\left(\Lambda_2\left(\Phi_2 \left(\Theta^A \otimes \id^{\tilde{A}} \right) \Phi_1  \otimes \id^{E_AB} \right)\Lambda_1\Delta\right) \nonumber \\
			=& \supr{\Lambda_i\in \mathcal{C}}\ E_{\mathcal{E}^{A\tilde{A}E_A|B}}\left(\Lambda_2\left(\Phi_2\otimes \id^{E_AB}\right) \left(\Theta^A \otimes \id^{\tilde{A}E_AB}  \right)\left(\Phi_1\otimes \id^{E_AB}\right)\Lambda_1\Delta\right) \nonumber \\
			\le& \supr{\Lambda_i\in \mathcal{C}}\ E_{\mathcal{E}^{A\tilde{A}E_A|B}}\left(\Lambda_2\left( \Theta^A \otimes \id^{\tilde{A}E_AB} \right)\Lambda_1\Delta\right) \nonumber \\
			=&\supr{\Lambda_i\in \mathcal{C}}\ E_{\mathcal{E}^{AE_A|B}}\left(\Lambda_2\left(\Theta^{A} \otimes \id^{E_AB}  \right)\Lambda_1\Delta\right) \nonumber \\
			=&C_{\mathcal{C}}^{\mathcal{E},E}\left(\Theta\right),
		\end{align}
		where we included $\tilde{A}$ into $E_A$ in the second last line.
		Also the proof that convexity can be inherited is straightforward. Assume that $E_{\mathcal{E}^{AE_A|B}}$ is convex. Then we find
		\begin{align}
			C_{\mathcal{C}}^{\mathcal{E},E}\left(\sum_k p_k \Theta_k\right)=&\supr{\Lambda_i\in \mathcal{C}}\ E_{\mathcal{E}^{AE_A|B}}\left(\Lambda_2\left(\sum_k p_k \Theta_k^{A} \otimes \id^{E_AB} \right)\Lambda_1\Delta\right) \nonumber \\
			=&\supr{\Lambda_i\in \mathcal{C}}\ E_{\mathcal{E}^{AE_A|B}}\left(\sum_k p_k \Lambda_2\left( \Theta_k^{A} \otimes \id^{E_AB} \right)\Lambda_1\Delta\right) \nonumber \\
			\le&\supr{\Lambda_i\in \mathcal{C}}\ \sum_k p_k\ E_{\mathcal{E}^{AE_A|B}}\left( \Lambda_2\left( \Theta_k^{A} \otimes \id^{E_AB} \right)\Lambda_1\Delta\right) \nonumber \\
			\le& \sum_k p_k\ \supr{\Lambda_i\in \mathcal{C}}\  E_{\mathcal{E}^{AE_A|B}}\left( \Lambda_2\left( \Theta_k^{A} \otimes \id^{E_AB} \right)\Lambda_1\Delta\right) \nonumber \\
			=& \sum_k p_k\ C_{\mathcal{C}}^{\mathcal{E},E}\left(\Theta_k\right),
		\end{align}
		i.e., $	C_{\mathcal{C}}^{\mathcal{E},E}$ is convex too.
		The statement about faithfulness in the case of $\rm MIO$ is a direct consequence of Cor.~\ref{cor:conversion}.
	\end{proof}
	
	\section{Comparison of the monotones}\label{sec:relMeasures}
	
	In this Section, we discuss how the different monotones introduced in the main text and the analog quantity from Eq.~(\ref{eq:LOCCrMeas}) bound each other. Since their definitions involve an infimum, the different sets $\mathcal{E}$ and $\mathcal{C}$ we are considering obey the inclusion relations shown in Fig.~1 in the main text, and $\rm LOCC_r \subset LOCC_{r+1} \subset LOCC$, for $\mathcal{M}\in \left\{\rm no, all\right\}$, we obtain 
	\begin{align}
		&E_{{\rm{LOCC_r}}^{A|B},D}^{\mathcal{S},\mathcal{M}}\left(\Theta\right) \ge E_{{\rm{LOCC_{r+1}}}^{A|B},D}^{\mathcal{S},\mathcal{M}}\left(\Theta\right) \ge E_{{\rm{LOCC}}^{A|B},D}^{\mathcal{S},\mathcal{M}}\left(\Theta\right)\ge E_{{\rm{\overline{LOCC}}},D}^{\mathcal{S},\mathcal{M}}\left(\Theta\right) \ge E_{{\rm{SEP}}^{A|B},D}^{\mathcal{S},\mathcal{M}}\left(\Theta\right), \nonumber \\
		&C_{{\rm SIO},D}^{\mathcal{S},\mathcal{M}}\left(\Theta\right) \ge C_{{\rm LOP},D}^{\mathcal{S},\mathcal{M}}\left(\Theta\right) \ge
		C_{{\rm IO},D}^{\mathcal{S},\mathcal{M}}\left(\Theta\right) \ge
		C_{{\rm MIO},D}^{\mathcal{S},\mathcal{M}}\left(\Theta\right),  \nonumber \\
		&C_{{\rm SIO},D}^{\mathcal{S},\mathcal{M}}\left(\Theta\right) \ge C_{{\rm DIO},D}^{\mathcal{S},\mathcal{M}}\left(\Theta\right) \ge
		C_{{\rm MIO},D}^{\mathcal{S},\mathcal{M}}\left(\Theta\right).  
	\end{align}
	For $\mathcal{M}=\rm free$, it is not straightforward to establish a similar relation, since smaller sets of operations (over which we perform infima) also potentially include fewer free destructive measurements (over which we take suprema).
	
	Within one set of operations $\mathcal{E}$ or $\mathcal{C}$, other inequalities emerge from the different choices of states and destructive measurements over which we optimize: the supremum over the set of all states/destructive measurements is not smaller than the supremum over the set of free states/destructive measurements. Therefore, we find
	\begin{align}\label{eq:inequMono1}
		E_{\mathcal{E}^{A|B},D}^{\mathcal{D},\mathcal{M}}\left(\Theta\right) &\ge E_{\mathcal{E}^{A|B},D}^{\mathcal{W},\mathcal{M}}\left(\Theta\right), \nonumber \\
		E_{\mathcal{E}^{A|B},D}^{\mathcal{S},{\rm all}}\left(\Theta\right) &\ge E_{\mathcal{E}^{A|B},D}^{\mathcal{S},{\rm free}}\left(\Theta\right), \nonumber \\
		C_{{\mathcal{C}},D}^{\mathcal{D},\mathcal{M}}\left(\Theta\right) &\ge C_{{\mathcal{C}},D}^{\mathcal{I},\mathcal{M}}\left(\Theta\right), \nonumber \\
		C_{{\mathcal{C}},D}^{\mathcal{S},{\rm all}}\left(\Theta\right) &\ge C_{{\mathcal{C}},D}^{\mathcal{S},{\rm free}}\left(\Theta\right).
	\end{align}
	Finally, since $D$ is contractive by assumption, we also have
	\begin{align}\label{eq:inequMono2}
		E_{\mathcal{E}^{A|B},D}^{\mathcal{S},{\rm no}}\left(\Theta\right) &\ge E_{\mathcal{E}^{A|B},D}^{\mathcal{S},{\rm all}}\left(\Theta\right), \nonumber \\
		C_{{\mathcal{C}},D}^{\mathcal{S},{\rm no}}\left(\Theta\right) &\ge C_{{\mathcal{C}},D}^{\mathcal{S},{\rm all}}\left(\Theta\right).
	\end{align}
	With the same arguments, the inequalities in Eqs.~(\ref{eq:inequMono1},\ref{eq:inequMono2}) also hold for $\mathcal{E}={\rm LOCC_r}$. 
	For the monotones with respect to operations in $\mathcal{E}$, we visualize these inequalities in Fig.~\ref{fig::connectionMeasures}
	
	\begin{figure*}[tb]
		\centering
		\includegraphics[width=\textwidth]{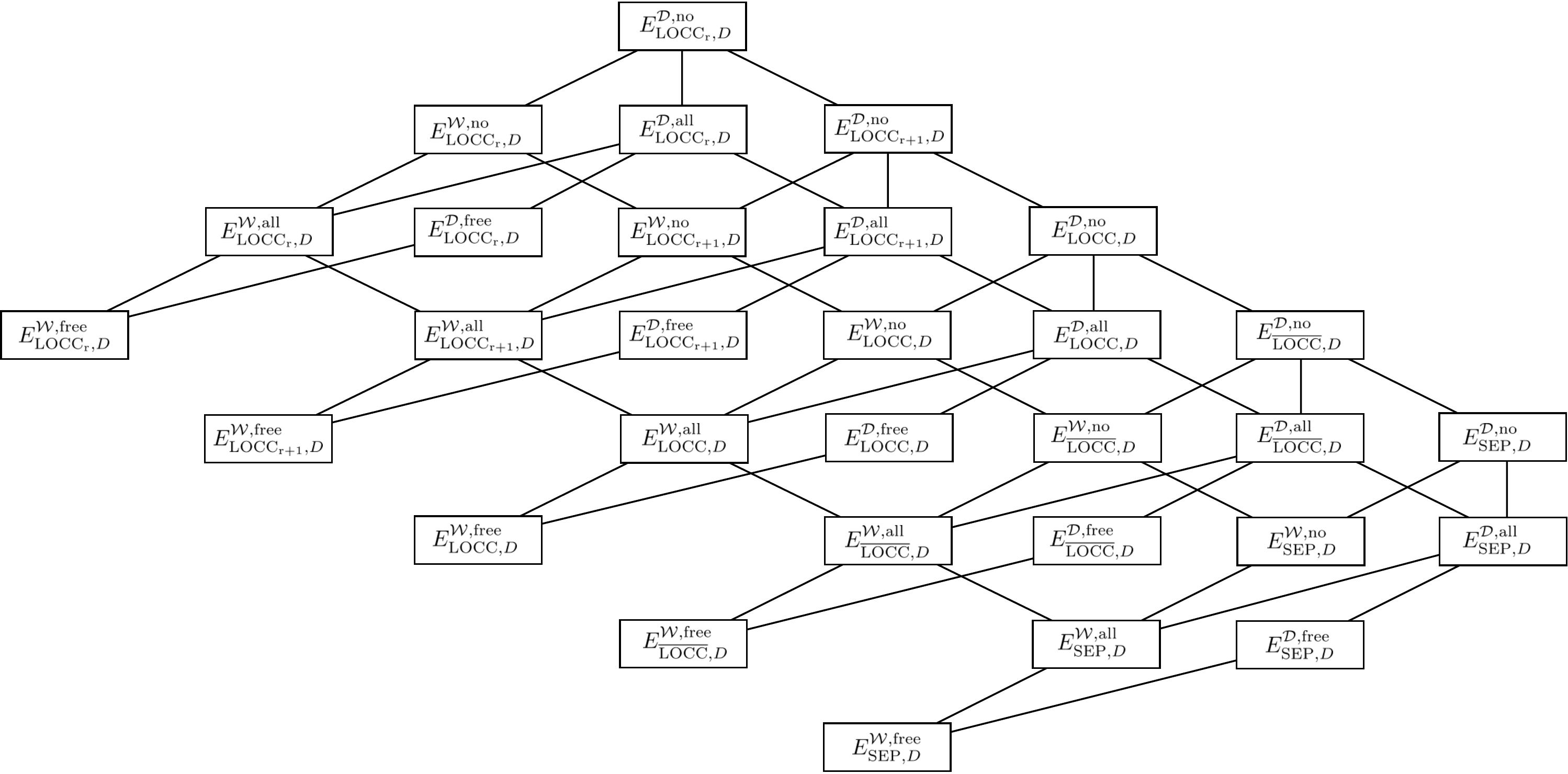}
		\caption{{\bf Bounding the entanglement monotones.} Visualization of the relation of the monotones quantifying the entanglement of operations and their relation to the quantity defined in Eq.~(\ref{eq:LOCCrMeas}). A monotone which is connected to a monotone below it is an upper bound to that one. For a more detailed discussion, see Sec.~\ref{sec:relMeasures}. }
		\label{fig::connectionMeasures}
	\end{figure*}

	\section{Reduction to resource theories of quantum states} \label{sec:reduction}
	In this Section, we discuss in detail how our findings include the results on resource theories of states presented in Ref.~\cite{Streltsov2015}. To this end, we consider preparation and replacement channels.  A replacement channel is a quantum operation with fixed output (which we will use as an index), i.e., $\Theta_\tau \rho=\tau \tr \rho$ for a fixed $\tau$.  For such an operation, we find with the help of Lem.~\ref{lem:obvious2} 
	\begin{align}
		C_{\mathcal{C},D}^{\mathcal{I},{\rm no}}\left(\Theta_\tau\right)=& \mini{\Lambda \in \mathcal{C}}\ \maxi{\sigma \in \mathcal{I}}\ D\left(\Theta_\tau\sigma,\Lambda\sigma \right) \nonumber\\
		=&\mini{\Lambda \in \mathcal{C}}\ \maxi{\sigma \in \mathcal{I}}\ D\left(\tau,\Lambda\sigma \right) \nonumber\\
		\ge& \mini{\sigma \in \mathcal{I}}\ D\left(\tau,\sigma \right),
	\end{align}
	since $\Lambda$ cannot create coherence from an incoherent state. This lower bound, however, can always be reached by the appropriate (free) replacement channel. Therefore, we find
	\begin{align}\label{eq:redToStateCoh}
		C_{\mathcal{C},D}^{\mathcal{I},{\rm no}}\left(\Theta_\tau\right)= \mini{\sigma \in \mathcal{I}}\ D\left(\tau,\sigma \right),
	\end{align}
	which was also shown in Ref.~\cite{Gour2019a} for the relative entropy. In addition, using Lem.~\ref{lem:obvious}, we find
	\begin{align}\label{eq:partRedToStateEnt}
		E_{{\mathcal{E}}^{A|B},D}^{\mathcal{W},{\rm no}}\left(\Phi^{AB}_2\left(\Theta^A_\tau\otimes \id^B\right)\Delta^{AB}\right) = &\infi{\Lambda \in {\mathcal{E}}^{A|B}}\  \maxi{\sigma \in \mathcal{W}_{A|B}} D\left(\Phi^{AB}_2\left(\tau^A\otimes \tr_A\left(\Delta^{AB} \sigma\right)\right) , \Lambda^{A|B} \sigma\right) \nonumber\\
		\ge&\infi{\Lambda \in {\mathcal{E}}^{A|B}}\   D\left(\Phi^{AB}_2\left(\tau^A\otimes  \ketbra{0}{0}^B\right) , \Lambda^{A|B}  \left(\ketbra{0}{0}^A\otimes \ketbra{0}{0}^B\right)\right) \nonumber\\
		=&  \mini{\sigma \in  \mathcal{W}_{A|B}}\ D\left(\Phi^{AB}_2\left(\tau^A\otimes  \ketbra{0}{0}^B\right) ,  \sigma \right),
	\end{align}
	which also holds for $\mathcal{E}= {\rm LOCC_{r\ge3}}$.
	From Thm.~\ref{thm:inequality} and Eqs.~(\ref{eq:redToStateCoh},\ref{eq:partRedToStateEnt}), we recover Thm.~1 of Ref.~\cite{Streltsov2015} as a special case, i.e.,
	\begin{align}
		\mini{\sigma \in \mathcal{I}}\ D\left(\tau,\sigma \right)=C_{\mathcal{C},D}^{\mathcal{I},{\rm no}}\left(\Theta_\tau\right)\ge E_{\mathcal{E}^{A|B},D}^{\mathcal{W},{\rm no}}\left(\Phi^{AB}_2\left(\Theta^A_\tau\otimes \id^B\right)\Delta^{AB}\right) \ge \mini{\sigma \in  \mathcal{W}_{A|B}}\ D\left(\Phi^{AB}_2\left(\tau^A\otimes \ketbra{0}{0}^B\right) , \sigma \right)
	\end{align}
	for all $\Phi_2\in\mathcal{C}$. Since Thm.~\ref{thm:inequality} also holds for $\mathcal{E}= {\rm LOCC_{r\ge3}}$ (which we noted below it), the above statement is also true for $\mathcal{E}= {\rm LOCC_{r\ge3}}$.
	
	Next assume that $\tilde{\mathcal{E}} \in \left\{\rm{LOCC_{r\ge3}}, \rm{\overline{LOCC}}, \rm{SEP} \right\}$. Since all of these sets are convex and closed, with the help of Thm.~2 of Ref.~\cite{Gour2019a} and Lem.~\ref{lem:obvious}, we note that
	\begin{align}\label{eq:upperBoundState}
		E_{{\tilde{\mathcal{E}}}^{A|B},S}^{\mathcal{W},{\rm no}}\left(\mathcal{U}_\mathrm{CNOT}^{AB}\left(\Theta^A_\tau\otimes \id^B\right)\Delta^{AB}\right) 
		=& \mini{\Lambda \in {\tilde{\mathcal{E}}}^{A|B}}\  \maxi{\sigma \in \mathcal{W}_{A|B}} S\left(\mathcal{U}_\mathrm{CNOT}^{AB}\left(\Theta^A_\tau\otimes \id^B\right)\Delta^{AB} \sigma, \Lambda^{A|B} \sigma   \right)  \nonumber\\
		=& \maxi{\sigma \in \mathcal{W}_{A|B}}\ \mini{\Lambda \in {\tilde{\mathcal{E}}}^{A|B}} S\left(\mathcal{U}_\mathrm{CNOT}^{AB}\left(\Theta^A_\tau\otimes \id^B\right)\Delta^{AB} \sigma, \Lambda^{A|B} \sigma   \right) \nonumber \\
		=& \maxi{\sigma \in \mathcal{I}}\ \mini{\rho \in \mathcal{W}_{A|B}} S\left(\mathcal{U}_\mathrm{CNOT}^{AB}\left(\tau^A\otimes \sigma^B\right), \rho   \right) \nonumber \\
		\le& \maxi{\Phi \in \mathcal{C}}\ \mini{\rho \in \mathcal{W}_{A|B}} S\left(\Phi^{AB}\left(\tau^A\otimes \ketbra{0}{0}^B\right), \rho   \right). 
	\end{align}
	Note that the same holds for all other divergences that satisfy Thm.~2 of Ref.~\cite{Gour2019a}, e.g., the trace distance. 
	From the discussion in Sec.~\ref{sec:relMeasures}, we remember that 
	\begin{align}
		E_{{\rm SEP}^{A|B},S}^{\mathcal{W},{\rm no}}\le E_{{\rm LOCC}^{A|B},S}^{\mathcal{W},{\rm no}} \le E_{{\rm LOCC_3}^{A|B},S}^{\mathcal{W},{\rm no}}
	\end{align}
	and therefore Eq.~(\ref{eq:upperBoundState}) is also valid for $\tilde{\mathcal{E}}={\rm LOCC}$. Using Eq.~(\ref{eq:partRedToStateEnt}), we obtain in addition
	\begin{align}
		\maxi{\Phi_1,\Phi_2 \in {\mathcal{C}}} E_{{\mathcal{E}}^{A|B},S}^{\mathcal{W},{\rm no}}\left(\Phi^{AB}_2\left(\Theta^A_\tau\otimes \id^B\right)\Phi^{AB}_1\Delta^{AB}\right)\ge \maxi{\Phi_2 \in {\mathcal{C}}}\ \mini{\sigma \in  \mathcal{W}_{A|B}}\ S\left(\Phi^{AB}_2\left(\tau^A\otimes \ketbra{0}{0}^B\right) , \sigma \right),
	\end{align}
	which, together with Thm.~\ref{thm:main}, allows us to conclude that 
	\begin{align}
		\maxi{\Phi_1,\Phi_2 \in {\mathcal{C}}} E_{{\mathcal{E}}^{A|B},S}^{\mathcal{W},{\rm no}}\left(\Phi^{AB}_2\left(\Theta^A_\tau\otimes \id^B\right)\Phi^{AB}_1\Delta^{AB}\right)=\maxi{\Phi_2 \in {\mathcal{C}}}\ \mini{\sigma \in  \mathcal{W}_{A|B}}\ S\left(\Phi^{AB}_2\left(\tau^A\otimes \ketbra{0}{0}^B\right) , \sigma \right).
	\end{align}
	Therefore, invoking again Thm.~\ref{thm:main}, this time together with Eq.~(\ref{eq:redToStateCoh}), using preparation channels, our results recover the corresponding findings on static resources presented in Eq.~(8) of Ref.~\cite{Streltsov2015}:
	\begin{align}
		\maxi{\Phi_2 \in {\rm IO}}\ \mini{\sigma \in  \mathcal{W}_{A|B}}\ S\left(\Phi^{AB}_2\left(\tau^A\otimes \ketbra{0}{0}^B\right) , \sigma \right)=\mini{\sigma \in \mathcal{I}}\ S\left(\tau,\sigma \right).
	\end{align}
	This also implies that  Cor.~\ref{cor:conversion} includes Thm.~2 of Ref.~\cite{Streltsov2015}: a state $\rho$ can be converted to an entangled state via IO if and only if $\rho$ is coherent. 
	
	All in all, the discussion in this section shows that our results on the quantitative connection of dynamic entanglement and coherence reduce to the the corresponding results on static entanglement and coherence if we consider preparation channels.

\end{document}